%% file: arxiv.tex
\documentclass[
  10pt,
  letterpaper,
  headings=small,
  numbers=noendperiod,
  bibliography=totoc
]{scrartcl}

\usepackage[
  left=1.00in,
  right=1.00in,
  top=0.86in,
  bottom=0.95in,
  includehead
]{geometry}
\usepackage[T1]{fontenc}
\usepackage[utf8]{inputenc}
\usepackage{textcomp}
\usepackage{microtype}
\usepackage[tt=false,type1=true]{libertine}

\usepackage{xurl}
\usepackage{graphicx}
\usepackage{amsfonts}
\usepackage{needspace}
\usepackage{algorithm}
\usepackage{tabularx}
\usepackage{booktabs}
\usepackage{algorithmicx}
\usepackage{stfloats}
\usepackage{dsfont}
\usepackage{enumitem}
\usepackage{todonotes}
\usepackage{blindtext}
\usepackage{natbib}
\setcitestyle{numbers,square,comma}
\usepackage{amsmath}
\usepackage{amsthm}
\usepackage[libertine]{newtxmath}
\newtheorem{proposition}{Proposition}
\newtheorem{lemma}{Lemma}
\newtheorem{theorem}{Theorem}

\newtheorem{example}{Example}

\renewcommand{\Pr}[1]{\mathds{P}\hspace{-0.1cm}\left[#1\right]}

\usepackage{balance} 

\usepackage{dsfont}
\usepackage{booktabs}
\usepackage{tabularx}
\usepackage{pgfplots}
\usepackage{enumitem}
\usepackage{todonotes}
\usepackage{blindtext}
\usepackage{adjustbox}
\usepgfplotslibrary{fillbetween}

\usepackage{subcaption}

\usepackage{xcolor}
\pgfplotsset{compat=1.18}

\usepackage{hyperref}
\hypersetup{hidelinks}
\definecolor{mygreen}{HTML}{1d9954}

\makeatletter
\renewcommand{\maketitle}{%
  \begingroup
  \thispagestyle{plain}%
  \raggedright
  {\Large\bfseries \@title\par}%
  \vspace{0.70em}%
  {\normalsize \@author\par}%
  \ifx\@date\@empty\else
    \vspace{0.25em}{\normalsize \@date\par}%
  \fi
  \vspace{0.80em}%
  \endgroup
}

\newcommand{\compacttoc}{%
  \par\vspace{0.65em}%
  \begin{center}{\sffamily\bfseries Contents}\end{center}%
  \vspace{-0.80em}%
  \@starttoc{toc}%
}
\makeatother

\renewenvironment{abstract}{%
  \par\vspace{.65em}%
  \begin{center}{\sffamily\bfseries Abstract}\end{center}%
  \vspace{-0.55em}%
  \small\noindent\ignorespaces%
}{\par\normalsize\vspace{.45em}}

\newenvironment{keywords}{%
  \par\smallskip\noindent{\small\bfseries Additional Key Words and Phrases: }\small\ignorespaces%
}{\par\normalsize\vspace{.45em}}

\title{Asymmetric Information in Dynamic Lending Games}

\date{\footnotesize{Department of Informatics, King's College London}}
\author{%
	Kristijan Atanasov, Frederik Mallmann-Trenn, and Carmine Ventre
}

\begin{document}
\maketitle

\begin{abstract}
\noindent We study a dynamic lending game with limited but reusable supply. Borrower arrivals and loan durations are stochastic. Demand is allocated cheapest-first, subject to available supply, so a lender's quote affects both its current payoff and its ability to serve future borrowers. Motivated by industry practice, we study this problem through the lens of information disclosure: Lenders observe either only their own available supply or also their rivals' aggregate available supply. We ask how providing this information to all, some or none of the lenders affects payoffs as demand varies relative to supply, and how these effects change when a fixed total supply is divided among different numbers of lenders. Interestingly, in the case of asymmetric information disclosure, total payoffs can increase, while total lending remains the same. We identify sufficient mathematical regimes in which this phenomenon occurs. Who receives supply information matters differently as demand rises. For the case of two lenders, informing only the larger lender yields higher total lender payoff at lower demand, while informing only the smaller lender can yield the highest total lender payoff at higher demand. For a platform deciding which lenders can observe rival supply, the results suggest that disclosure can change the fees paid by borrowers without changing how much is lent. Wider access to supply information need not lower those fees: its effect depends on who receives the information and the balance between demand and supply.
\end{abstract}


\setcounter{tocdepth}{2}
\tableofcontents

\newpage

\section{Introduction} 
Consider a borrower seeking one unit of an asset from competing lenders. Each lender quotes a fee, and the lowest available offer wins. The winning lender receives its fee but cannot lend that unit again until it returns. Winning therefore creates a strategic trade-off: it generates an immediate payoff while temporarily reducing the lender's ability to compete for subsequent demand. Motivating applications with reusable supply include securities lending and rentals of accommodation, construction equipment and computing capacity~\cite{ecb2016securities,FTS2026}.

We model this interaction as a dynamic lending game with stochastic borrower arrivals and loan returns. Ties at the lowest fee are broken uniformly at random, with the selected lender supplying the unit. Each lender's strategy specifies its quote from the supply information it observes. If no lender has an available unit, the arriving demand is not served. Total supply remains fixed, but the amount and distribution of available supply continually change.

Competing suppliers need not have equal access to market information, particularly when such information is sold through tiered subscription services. In rental markets, AirDNA is one example: the market intelligence available to users varies across its subscription plans~\cite{airdna2026plans}. In securities lending, DataLend has used a ``give-to-get'' model, providing anonymised, aggregated industry data in return for clients' transaction data~\cite{datalend2013give}. Its inputs include aggregated outstanding loan and supply information~\cite{ross2019datalend}. The relevant question is not only how many lenders receive market information, but how information access is distributed relative to the ownership of supply. We examine how this interaction shapes pricing as demand pressure changes.

Borrower arrival and loan return rates determine market tightness, while lenders may differ in what they observe about rivals’ available supply. Taking information access as given, we compare settings in which all, some, or none of the lenders observe their rivals' aggregate available supply. We ask how the allocation of this information affects prices and payoffs as demand pressure, the division of a fixed total supply and the permitted fee set vary. In the case of two lenders, reassigning supply information from the larger lender to the smaller can even increase both lenders' payoffs, including that of the lender that loses access to the information. These gains arise without an increase in aggregate lending.

\smallskip \noindent \textbf{Our contribution.}
Our results concern observation-constrained pure stationary equilibria. When multiple equilibria exist, we select the equilibrium with the highest total discounted lender payoff.

\noindent\emph{Notation.} We list lenders in a fixed order. The letter $O$ means that a lender observes only its own available supply, while $A$ means that it also observes the aggregate available supply of its rivals. We refer to an $O$-lender as uninformed and an $A$-lender as informed. For two lenders with unequal supplies, we list the larger lender first. Hence, $OO$, $AO$, $OA$ and $AA$ mean that neither lender, only the larger lender, only the smaller lender and both lenders observe rival supply, respectively. The supply vector $(x,y)$, for $x,y\in \mathbb{N}$, 
means that the larger lender owns $x$ units and the smaller lender owns $y$ units.
\begin{enumerate}[nosep, left=0pt]

\item For two lenders with supplies $(2,1)$ and three permitted fees, $AO$ yields higher total lender payoff than $OA$ when borrower demand is below supply. When demand is close to or equal to the supply, with all other parameters unchanged, the ranking reverses: $OA$ yields the highest total lender payoff among all four information regimes. For the same supplies in a two-fee setting, we derive an explicit sufficient parameter region in which $OA$ has higher payoff than $AO$ and numerically map how the $AO$--$OA$ payoff difference varies with borrower arrival and return rates (Example~\ref{prop_supply_info}
and Proposition~\ref{prop:two-fee-region}; see also Figure~\ref{fig:payoff-map-regimes}).
\item For two lenders with supplies $(m,1)$, with any integer $m\geq2$, and two permitted fees whose lower fee exceeds half the higher, we establish explicit sufficient conditions 
under which $OA$ yields higher total payoff than $AO$. These conditions define a nonempty open parameter region
(Theorem~\ref{thm:general-larger-supply}).

\item The $OA$ advantage need not come at the expense of the lender that loses information. With supplies $(3,1)$, moving from $AO$ to $OA$ strictly increases both lenders' equilibrium payoffs, including that of the larger lender that loses access to rival-supply information. The reassignment is therefore a strict Pareto improvement for the two lenders (Proposition~\ref{prop:pareto-gain}).

\item Splitting an informed lender into two independent informed lenders can raise total payoff while leaving total supply unchanged (Proposition~\ref{prop:split-informed-lender}).
\end{enumerate}

\noindent Within each comparison at fixed total supply and fixed arrival, return and discount rates, aggregate lending is unchanged (Lemma~\ref{lem:lending-invariance}). The payoff differences therefore arise from fees and the allocation of loans across lenders and supply states, rather than additional lending.

\input{figure_arxiv.tex}

\smallskip \noindent \textbf{Related literature.}
Our work is closely related to dynamic pricing with reusable supply. \citet{FTS2026} study competition with reusable supply and logit demand. They compare competitive pricing when all firms condition only on their own occupancy or all condition on the full occupancy state, alongside a joint-pricing benchmark. Their numerical results show that sharing occupancy information can intensify competition. \citet{RQA2018} analyse competing providers with stochastic arrivals and departures and compute approximate Nash equilibria. Related models examine competition with supplies depleted by sales \cite{MT2011,BHAW2024,dudey1992dynamic}, demand learning by a single seller \cite{BDKS2015}, and static versus dynamic pricing for reusable resources \cite{ES2025}. We examine how the assignment of supply information interacts with unequal supplies and demand pressure. Every arrival is served whenever supply is available, and loans return independently at a common, fee-independent rate. Total lending is consequently unchanged by pricing, information and the division of supply, allowing us to isolate changes in fees and individual allocations.

\citet{GJS2023} study dynamic price competition with private supplies and show that mutual capacity-information exchange raises firms’ profits relative to private capacities. Unilateral disclosure can also raise profits and may outperform mutual exchange. Their analysis incorporates beliefs about rivals' remaining supply. \citet{SR2019} infer hidden rival inventories in finite-horizon price competition, while \citet{AFJP2020} examine inventory information sharing in dynamic auctions. We share their interest in how information affects future competition, but study recurring available supply and restrict pricing rules to current supply observations, rather than inference from the history of prices and transactions.

The consequences of disclosure depend on the information provided and the form of competition \cite{Raith1996}. \citet{GJNBN2026} show that selective provision of demand information can benefit competing sellers, while \citet{CS2018} study aggregate inventory disclosure to customers. Our signal reports total available rival supply, without identifying each rival's holdings when there are more than two lenders. We show that demand can reverse which recipient generates the higher total payoff, and that informing only the smaller lender can strictly benefit both lenders relative to informing only the larger. These results show how the allocation of supply information and the ownership of reusable units jointly shape competition.

Our work also connects to algorithmic pricing and learning. \citet{CCDP2020} show that independent Q-learning agents can learn high-price outcomes supported by reward--punishment behaviour. \citet{FPR2025} study learning in episodic markets with limited inventories, while \citet{CCP2022} examine the role of tick size in electronic-market competition. We compute equilibria in observation-constrained pure stationary strategies, consistent with the distinction between restricted and unrestricted strategy sets \cite{BV2004}. We check complete unilateral deviations within that class without specifying a learning algorithm. The payoff effects arise from supply information and pricing incentives, rather than particular exploration or updating rules; higher fees are not by themselves identified with learned collusion.

\section{Model}\label{sec:model}
We study a non-cooperative lending game in continuous time over an infinite horizon. Borrowers arrive according to a Poisson process with rate $\lambda>0$, each requesting one unit. Each loan has an independent exponentially distributed duration with rate $\mu>0$, giving an expected duration of $\frac1\mu$.\footnote{Of the applications mentioned above, uncertain end times are natural for securities lending and computing tasks. Accommodation and equipment rentals often have scheduled end dates, so the exponential-duration assumption is a simplification in those settings. It means that the rate at which units return depends only on how many are currently on loan, rather than on how long they have already been out.} Lenders compete by quoting fees, and each seeks to maximise its expected discounted payoff. We consider pure stationary strategies, which assign a quote to each observation of available supply. Lenders always observe their own available supply, but may also observe the aggregate available supply of their rivals. The allocation of this information determines which supply conditions each lender can distinguish when choosing its quote.  

\subsection{Lenders and supply}

Let $I=\{1,\ldots,n\}$ be the set of lenders, where $n\geq 2$ is the number of lenders. Each lender $i\in I$ owns $c_i\in\mathbb{N}$ identical, indivisible units, giving total supply $C=\sum_{i\in I}c_i$. The supply is a fixed parameter of each game. We distinguish a lender's supply, meaning the total number of units it owns, from its available supply, meaning the units it can currently lend. 

At time $t\geq 0$, let $x_i(t)$ denote lender $i$'s available supply. Each unit is either available or on loan, so $x_i(t)\in\{0,\ldots,c_i\}$ and $c_i-x_i(t)$ units are on loan. A lent unit remains unavailable until its loan ends, at which point it returns to the same lender. A lender with no available units cannot serve an arriving borrower. The supply state and its state space are
\[
s(t)=(x_1(t),\ldots,x_n(t))\in S,
\qquad
S=\prod_{i\in I}\{0,\ldots,c_i\}.
\]
Let $s^0=(c_1,\ldots,c_n)$ denote the starting fully available state. All equilibrium and payoff comparisons below are evaluated from $s^0$. Additionally, for each $k\in\{0,\ldots,C\}$, let $S_k$ denote the set of states with $k$ units available in total.

\noindent Lenders need not observe this full state; their information is specified below through a coarser observation of $s(t)$.

\subsection{Lending mechanism}
Let $P\subset(0,\infty)$ be the finite, nonempty set of permitted fees, common to all lenders. For a state $s=(x_1,\ldots,x_n)\in S$, let $o_i(s)$ denote lender $i$'s supply observation:
\[
o_i(s)=
\begin{cases}
x_i,
& \text{own-supply observation (O)},\\[4pt]
\left(x_i,\displaystyle\sum_{j\in I\setminus\{i\}}x_j\right),
& \text{aggregate rival-supply observation (A)}.
\end{cases}
\]
A pure stationary strategy $\sigma_i$ assigns one fee in $P$ to each observation at which lender $i$ has available supply. Fix a strategy profile $\sigma=(\sigma_i)_{i\in I}$, containing one strategy for each lender. The fee quoted by lender $i$ in state $s$ is $p_i(s)=\sigma_i(o_i(s))$ whenever $x_i>0$. The strategy depends only on that current observation: it cannot condition on time or on the history of past observations. This keeps the analysis focused on the economic effects of access to rival-supply information, rather than on learning about rival supply from market history. The available lenders are $\mathcal A(s)=\{i\in I:x_i>0\}$. When this set is nonempty, the lenders offering the lowest fee are
\[
\mathcal W(s,\sigma)
=
\left\{
i\in\mathcal A(s):
p_i(s)=\min_{j\in\mathcal A(s)}p_j(s)
\right\}.
\]
One lender is selected uniformly from $\mathcal W(s,\sigma)$. It supplies the entire requested unit and receives its quoted fee at booking. The probability that lender $i$ receives an arriving borrower is therefore
\[
a_i(s,\sigma)
=
\begin{cases}
\dfrac{1}{|\mathcal W(s,\sigma)|},
& i\in\mathcal W(s,\sigma),\\[6pt]
0,
& \text{otherwise},
\end{cases}
\]
where $|\mathcal W(s,\sigma)|$ is the number of tied lowest-fee lenders. All permitted fees are acceptable to borrowers, so every arrival is served whenever supply is available. If $\mathcal A(s)=\varnothing$, we set $\mathcal W(s,\sigma)=\varnothing$ and $a_i(s,\sigma)=0$ for every lender; the borrower leaves without a loan.
A booking reduces the selected lender's available supply by one. When the loan ends, the unit returns to that lender, increasing its available supply by one. Either event changes the supply state from $s$ to some new state $s'$.

\subsection{State transitions}
For fixed lender strategies $\sigma$, the supply state follows a continuous-time Markov chain on $S$. The exponential loan durations mean that return rates depend on the number of units currently on loan, rather than how long those loans have been outstanding. Let $e_i$ denote the vector with a one in coordinate $i$ and zeros elsewhere. From state $s=(x_1,\ldots,x_n)$, a booking by lender $i$ moves the state to $s-e_i$, while a return to that lender moves it to $s+e_i$. Denote the transition-rate matrix by $Q(\sigma)$, where $Q_{s,s'}(\sigma)$ is the rate of moving from state $s$ to a distinct state $s'$. Its booking and return entries are
\begin{align*}
Q_{s,s-e_i}(\sigma)
&= \lambda a_i(s,\sigma),
&& x_i>0,\\
Q_{s,s+e_i}(\sigma)
&= \mu(c_i-x_i),
&& x_i<c_i.
\end{align*}
The booking rate is the borrower arrival rate $\lambda$ multiplied by $a_i$, the probability that lender $i$ receives the allocation. The total return rate for a lender is the per-unit return rate multiplied by the number of outstanding units $c_i-x_i$.
All other off-diagonal entries are zero: each event changes the available supply of only one lender by one unit. The diagonal is the negative total rate of leaving the state,
\[
Q_{s,s}(\sigma)
=
-\sum_{\substack{s'\in S\\s'\neq s}}Q_{s,s'}(\sigma).
\]
An arrival finding no available supply leaves the state unchanged and is not included in this exit rate. Returns can still occur, allowing lenders to serve subsequent arrivals.
Under the fixed strategy profile $\sigma$, the probability of being in state $s'$ after an interval of length $h\to0^+$, conditional on starting in state $s$, is
\[
\Pr{s(t+h)=s' \mid s(t)=s}
\approx
\begin{cases}
hQ_{s,s'}(\sigma),
&  s'\neq s,\\[4pt]
1-h\sum_{\substack{\tilde s\in S\\\tilde s\neq s}}Q_{s,\tilde s}(\sigma)
& s'=s.
\end{cases}
\]

\subsection{Payoffs and state values}
A lender receives its quoted fee once, at booking. For a fixed strategy profile $\sigma$, lender $i$'s expected reward rate in state $s=(x_1,\ldots,x_n)$ is
\[
r_i(s,\sigma)
=
\begin{cases}
\lambda a_i(s,\sigma)p_i(s),
& x_i>0,\\[4pt]
0,
& x_i=0.
\end{cases}
\]
Let $\delta>0$ denote the continuous-time discount rate, so that a payment received at time $t$ is multiplied by the discount factor $e^{-\delta t}$. Write $V_i(s,\sigma)$ for lender $i$'s expected discounted payoff over the remaining horizon, starting from state $s$ and following the strategy profile $\sigma$. Using the expected reward rate, this value is
\[
V_i(s,\sigma)
=
\mathbb{E}\left[
\int_0^\infty e^{-\delta t}r_i(s(t),\sigma)\,\mathrm{d}t
\,\middle|\,
s(0)=s
\right],
\]
where $s$ specifies the initial state, $s(t)$ the state reached at elapsed time $t$ and the expectation is over borrower arrivals, loan returns, and random allocation among tied lenders.
Fix an ordering of the states in $S$ and define the column vectors
\[
\mathbf r_i(\sigma) := \bigl(r_i(s,\sigma)\bigr)_{s\in S} \text{ and } \mathbf V_i(\sigma):= \bigl(V_i(s,\sigma)\bigr)_{s\in S} \text{.}
\]
For fixed strategies, the payoff vector is the unique solution to
\begin{equation}
\label{eq:payoff-system}
\left[\delta\mathrm{Id}-Q(\sigma)\right]
\mathbf V_i(\sigma)
=
\mathbf r_i(\sigma),
\end{equation}
where $\mathrm{Id}$ is the identity matrix of dimension $|S|$. For each $s\in S$, define the expected discounted time spent in state $s$ under $\sigma$ by
\[
z_s(\sigma) = \int_0^\infty e^{-\delta t} \Pr{s(t)=s\mid s(0)=s^0}\,\mathrm dt.
\]
Define $\mathbf z(\sigma) = \bigl(z_s(\sigma)\bigr)_{s\in S},$ using the same state ordering as $Q(\sigma)$, and let $\mathbf b^0$ be the column vector with entry $1$ at state $s^0$ and $0$ elsewhere.
Then
\begin{equation}
\label{eq:state-time-system}
\left[
\delta\mathrm{Id}-Q(\sigma)^\top
\right]
\mathbf z(\sigma)
=
\mathbf b^0.
\end{equation}
The payoff can equivalently be written as $V_i(s^0,\sigma) = \mathbf z(\sigma)^\top\mathbf r_i(\sigma)$.

\subsection{Equilibrium}

For each lender $i\in I$, let $\Sigma_i$ denote the set of admissible pure stationary strategies: each strategy assigns one fee in $P$ to every supply observation at which that lender has an available unit. The information available to each lender is fixed by the information regime. Write $\sigma_{-i}$ for the collection of strategies used by all lenders other than $i$.
A strategy profile $\sigma^\ast=(\sigma_i^\ast)_{i\in I}$, with $\sigma_i^\ast\in\Sigma_i$, is a Nash equilibrium within this strategy class if
\[
V_i(s^0,\sigma^\ast) \geq V_i\bigl(s^0,(\widetilde{\sigma}_i,\sigma^\ast_{-i})\bigr) \qquad \text{for every }i\in I \text{ and }\widetilde{\sigma}_i\in\Sigma_i,
\]
where $(\widetilde{\sigma}_i,\sigma^\ast_{-i})$ is the profile obtained by replacing lender $i$'s strategy while leaving its rivals' strategies unchanged.
A deviation may change the lender's quotes at any or all of its permitted observations, but cannot change the supply information it observes. We then recalculate its payoff from $s^0$, including any changes in which lender serves later borrowers, the available supply after each booking and return, and the fees received thereafter. The equilibrium condition therefore compares each alternative strategy over the current reward and all subsequent bookings and returns.

Pure stationary equilibria need not exist or be unique. For a fixed game and information regime $R$, let $\mathcal{E}^R$ denote the set of pure stationary equilibria. Whenever $\mathcal{E}^R$ is nonempty, we select the one with the highest total discounted payoff:
\[
\pi^R = \max_{\sigma\in\mathcal{E}^R} \sum_{i\in I} V_i(s^0,\sigma).
\]
\noindent\emph{Sole informed lender.}\label{c:sole_lender} Because every supply state is reachable from $s^0$, an informed lender that observes zero aggregate rival available supply must quote $\max P$ in equilibrium. Raising a lower fee at that observation changes neither its allocation probability nor any transition rate and strictly increases its payoff. This conclusion need not hold for an uninformed lender, which cannot condition its fee on rival available supply.

\section{Results}\label{sec:results}
We examine how supply information affects lender payoffs, combining analytical results for reduced games with numerical comparisons across supply and demand configurations. We first establish the invariance of total lending, then examine how information assignment and the division of supply affect individual and total lender payoffs.
\begin{lemma}[Total lending invariance]\label{lem:lending-invariance}
Fix total supply $C$, the borrower arrival rate $\lambda$, the return rate per outstanding unit $\mu$, and the discount rate $\delta$. From $s^0$, total available supply $X(t)=\sum_{i\in I}x_i(t)$ has the same distribution at every time under every admissible strategy, information regime and division of supply among lenders. Consequently, for each $k\in\{0,\ldots,C\}$, the expected discounted time with $k$ units available is unchanged, as is the expected discounted number of successful bookings. In particular,
\[
B(\sigma) = \lambda\sum_{k=1}^{C}\sum_{s\in S_k}z_s(\sigma)
\]
is the expected discounted number of successful bookings and has the same value under every admissible strategy, information regime and division of supply.
\end{lemma}

\begin{proof}
Fix the lenders' strategies $\sigma$ and a physical state $s=(x_1,\ldots,x_n)$ with total available supply $k=\sum_{i\in I}x_i$. Every arrival is served whenever $k>0$, and each outstanding unit returns independently at rate $\mu$. The aggregate transition rates are therefore
\[
\begin{aligned}
k &\longrightarrow k-1 &&\text{at rate } \sum_{i\in I}\lambda a_i(s,\sigma)=\lambda, &&k>0,\\
k &\longrightarrow k+1 &&\text{at rate } \sum_{i\in I}\mu(c_i-x_i)=\mu(C-k), &&k<C.
\end{aligned}
\]
These rates are identical for all physical states with the same total available supply $k$. Thus $X(t)$ is a continuous-time Markov chain whose transition rates and initial condition $X(0)=C$ are independent of strategies, information and the division of supply. Grouping states by total available supply gives
\[
\sum_{s\in S_k}z_s(\sigma)
=
\int_0^\infty e^{-\delta t}
\Pr{X(t)=k\mid X(0)=C}\,\mathrm dt.
\]
Individual state times may differ, but their sum over $S_k$ is unchanged. Successful bookings occur at rate $\lambda$ whenever $X(t)>0$, giving the expression for $B(\sigma)$ and establishing its invariance.
\end{proof}
\noindent The lemma rules out increased aggregate lending as an explanation for higher payoff, directing attention instead to the fees received and the allocation of loans across lenders and supply states.

\subsection{Information assignment}\label{sec_res_supply_info}

In this section, we focus on the case $n=2$ and ask whether the same lender should receive supply information as demand changes. We hold supplies, fees and the return rate fixed, and compare information arrangements at different arrival rates. The following result shows that informing the larger lender can generate higher payoffs in regimes with lower arrival rates, while informing only the smaller lender can produce the highest payoff in regimes with higher arrival rates.

We distinguish three demand regimes based on the relationship
between borrower demand and lending supply. Let
$c_{\max}=\max_{i\in I}c_i$ denote the largest individual
lender's supply, and let $C=\sum_{i\in I}c_i$ denote total
supply. In the \emph{Bertrand region}, $\lambda\leq c_{\max}\mu$,
at least one lender has sufficient supply to accommodate
demand on average, so competition resembles unconstrained
Bertrand competition. In the \emph{Edgeworth region},
$c_{\max}\mu<\lambda\leq C\mu$, no individual lender has
sufficient supply, but aggregate supply is sufficient,
making residual demand relevant. Motivated by
Bertrand--Edgeworth competition, we call this the
\emph{Edgeworth region}. Finally, in the \emph{shortage
region}, $\lambda>C\mu$, even aggregate capacity is
insufficient to accommodate demand on average.

These regions are illustrated in
Figure~\ref{fig:main-regime-map}, which also shows
how the relative total payoffs under $OA$ and $AO$
depend on the arrival and return rates.
The regime boundaries provide economic benchmarks
rather than equilibrium boundaries.

\begin{example}[The recipient of supply information]\label{prop_supply_info}
There exist supplies $c_1>c_2$, a set $P$ of three permitted fees, a return rate $\mu>0$, a discount rate $\delta>0$, and arrival rates $0<\lambda_L<\lambda_H$ such that equilibria exist under all four information regimes and
\[ \pi^{AO}(\lambda_L)>\pi^{OA}(\lambda_L), \]
whereas
\[
\pi^{OA}(\lambda_H)> \max\left\{ \pi^{OO}(\lambda_H),\, \pi^{AO}(\lambda_H),\, \pi^{AA}(\lambda_H)\right\}.
\]
\end{example}

\begin{proof}[Proof sketch]
Choose supplies $(c_1,c_2)=(2,1)$, permitted fees
$P=\{0.75,0.875,1\}$, return rate $\mu=1$, discount rate
$\delta=0.1$, initial state $s^0=(2,1)$ and $\lambda_L=1$.
Appendix~\ref{app:lambda-one} establishes the unique equilibria
under $AO$ and $OA$. In each regime, both lenders quote
$\frac{3}{4}$ when competing. When alone, the informed lender
quotes 1 and the uninformed lender quotes $\frac{3}{4}$. The regimes induce the same transition-rate matrix and differ only
in the fee charged when a lender is alone. Write $z_{ab}$ for their
common discounted time in state $(a,b)$. Relative to $OA$, the fee
received under $AO$ is $\frac{1}{4}$ more per booking in $(1,0)$
and $(2,0)$, and $\frac{1}{4}$ less in $(0,1)$. Fees coincide when
both lenders are available. Appendix~\ref{app:example-payoffs}
shows that $z_{10}>z_{01}$, so the larger lender spends more
discounted time alone than the smaller lender. With $\lambda=1$,
\[
\pi^{AO}(1)-\pi^{OA}(1)
=
\frac{1}{4}\bigl(z_{10}+z_{20}-z_{01}\bigr)>0.
\]
Thus, the higher fee is received over more discounted time under
$AO$, although the discounted time in each state is unchanged
between the regimes.

\par\smallskip
\noindent Set $\lambda_H=3$, retaining the other parameters.
Appendix~\ref{app:lambda-three} establishes unique equilibria
under $OO$, $AO$ and $OA$, with fees
\begin{center}
\small
\setlength{\tabcolsep}{4pt}
\renewcommand{\arraystretch}{1.05}
\begin{tabularx}{\linewidth}{@{}cXX@{}}
\toprule
Regime & Lender 1 & Lender 2 \\
\midrule
$\mathrm{OO}$
    & $1$ whenever available
    & $0.875$ whenever available \\

$\mathrm{AO}$
    & $1$ when alone; $0.875$ otherwise
    & $0.875$ whenever available \\

$\mathrm{OA}$
    & $1$ whenever available
    & $1$ when alone; $0.875$ otherwise \\
\bottomrule
\end{tabularx}
\end{center}
Under $AA$, there are two equilibria. Both lenders quote 1 when
alone and at most 0.875 when competing. Under $OO$ and $OA$,
lender 2 quotes less than lender 1 whenever both are available, and
therefore receives the borrower with probability 1. Under $AO$ and
both $AA$ equilibria, their competitive fees tie, so each lender
receives the borrower with probability $\frac{1}{2}$.

\noindent For $R\in\{AO,OA\}$, let $\sigma^R$ denote the equilibrium
at this arrival rate, and write
$z_{ab}^{R}:=z_{(a,b)}(\sigma^R)$.
Appendix~\ref{app:example-payoffs} shows that
$z_{11}^{OA}<z_{11}^{AO}$, while $z_{21}^{OA}=z_{21}^{AO}$.
The only states with two units available are $(1,1)$ and $(2,0)$.
By Lemma~\ref{lem:lending-invariance}, their combined discounted
time is unchanged, so $OA$ spends more discounted time in $(2,0)$.
In that state, lender 1 is alone and receives a fee of 1. In
$(1,1)$, both lenders are available and the booking fee is 0.875
under both regimes.

\noindent Under $OA$, lender 2 also charges 1 when alone, rather
than 0.875 under $AO$. Since the expected discounted number of
successful bookings is unchanged and $\lambda=3$, these fee and
state-time differences give
\[
\pi^{OA}(3)-\pi^{AO}(3)
=
\frac{3}{8}
\bigl[z_{01}^{AO}+(z_{11}^{AO}-z_{11}^{OA})\bigr]>0.
\]
The first term in brackets reflects the higher fee when lender 2
is alone. The second reflects less discounted time in $(1,1)$,
not a higher booking fee when both lenders are available.

\noindent Under $OO$, the discounted state times are the same as
under $OA$, but lender 2 charges 0.875 when alone. Both $AA$
equilibria have the discounted state times of $AO$ and booking
fees no higher than those under $OA$. The complete comparisons
in Appendix~\ref{app:example-payoffs} therefore give
\[
\pi^{OA}(3)>
\max\{\pi^{OO}(3),\pi^{AO}(3),\pi^{AA}(3)\}.
\]
\end{proof}

\noindent The larger lender controls more supply, yet informing the smaller lender can produce higher total payoff. We next identify conditions on the borrower arrival rate $\lambda$, the return rate per outstanding unit $\mu$, and the discount rate $\delta$ under which this advantage persists beyond the preceding result. The pricing incentives depend on the trade-off between receiving a lower fee now and waiting for a higher fee later. Define
\[
\beta=\frac{\lambda}{\lambda+\delta}.
\]
This is the expected discount factor until the next borrower arrives: a fee $p_H$ received at that arrival has expected present value $\beta p_H$. More frequent arrivals make waiting less costly, while stronger discounting makes it more costly.

\begin{theorem}[Informing the smaller lender] \label{thm:general-larger-supply}
For every integer $m\ge2$ and permitted fee set $P=\{p_L,p_H\}$ with $0<\frac{p_H}{2}<p_L<p_H$, there is a nonempty open set of positive borrower arrival rates $\lambda$, return rates $\mu$, and discount rates $\delta$ for which, with supplies $(c_1,c_2)=(m,1)$, AO and OA admit pure stationary equilibria, and
\[
\pi^{OA}(\lambda)>\pi^{AO}(\lambda).
\]
\end{theorem}

\begin{proof}
\noindent\textbf{1. Pricing incentives when units do not return.}
We construct the required parameter region by first calculating
pricing incentives when units do not return, and then choosing
positive values of $\mu$ that preserve these incentives.
Temporarily set $\mu=0$. Available supply can then only fall:
once lender 1 has fewer than $m$ units available, it cannot regain
all $m$. Its quote at $x_1=m$ therefore cannot affect its payoff
starting with fewer than $m$ available units.

\noindent Let $\alpha=\frac{p_L}{p_H}$. To calculate the benefit of matching a high
quote, consider a borrower arriving when lender 2 has its unit
available and lender 1 has $j$ available units, with $1\le j\le m$.
For this calculation, fix lender 1's strategy to quote $p_H$
whenever available. Undercutting at $p_L$ gives lender 2 that fee
immediately and ends its lending, since its only unit does not
return. Alternatively, lender 2 can continue quoting $p_H$ until
its unit is lent.

\noindent Write $p_Hg_j(\beta)$ for lender 2's expected discounted payoff
from this second choice, measured from the current arrival.
With probability $\frac{1}{2}$, lender 2 receives $p_H$ immediately.
Otherwise, lender 1 receives the borrower, its available supply
falls to $j-1$, and lender 2 waits for the next arrival. Hence
\[
g_0(\beta)=1,
\qquad
g_j(\beta)=\frac{1}{2}
          +\frac{\beta}{2}g_{j-1}(\beta),
\qquad j=1,\ldots,m.
\]
The boundary $g_0(\beta)=1$ applies when lender 1 has no available
unit, so lender 2 receives the current borrower at $p_H$.
Lender 2's payoff gain from continuing at $p_H$ rather than
undercutting is therefore
\[
p_Hg_j(\beta)-p_L
=
p_H\bigl(g_j(\beta)-\alpha\bigr).
\]

\noindent For the construction, we choose $\lambda$ and $\delta$ so that
undercutting is preferable in this calculation when $j=m$, whereas
continuing at $p_H$ is preferable for every $1\le j<m$.
The functions $g_j(\beta)$ decrease with $j$, as shown in
Appendix~\ref{app:theorem1}. It therefore suffices to choose
\[
g_m(\beta)<\alpha<g_{m-1}(\beta).
\]
Indeed, $g_j(\beta)\ge g_{m-1}(\beta)>\alpha$ for every $1\le j<m$, so checking $m-1$ covers every smaller positive rival available supply. We establish at the end of the proof that such choices exist for every supply and fee set in the theorem.

\noindent\textbf{2. Best responses under AO.}
Continue with $\mu=0$. Under AO, lender 2 observes only its own
available supply and owns one unit. Its two admissible strategies
are therefore to quote $p_L$ whenever available or to quote $p_H$
whenever available. We compute lender 1's best response to each
strategy, then check lender 2's choice against each resulting
lender-1 strategy.

\noindent Appendix~\ref{app:theorem1} establishes that, against lender 2's
constant quote $p_L$, lender 1 maximises its payoff by matching
$p_L$ whenever both lenders have available units. Against lender
2's constant quote $p_H$, lender 1 instead matches $p_H$ when
$1\le x_1<m$ and undercuts at $p_L$ when $x_1=m$. These latter
choices follow from the signs of $g_{x_1}(\beta)-\alpha$ established
in Step~1; the appendix derives the corresponding payoff gain
for lender 1. In both responses, lender 1 quotes $p_H$ whenever
it is the only lender with available units.

\noindent Holding each of these lender-1 strategies fixed, the appendix
evaluates lender 2's payoff from $s^0=(m,1)$ under its two
constant quotes. In both cases, lender 2 earns strictly more by
quoting $p_L$. We therefore retain lender 2's constant quote
$p_L$ and lender 1's response of matching when competing and
quoting $p_H$ when alone. Step~4 establishes that these strategies
form the unique AO equilibrium for the positive values of $\mu$
chosen there.

\noindent\textbf{3. Best responses under OA.}
Continue with $\mu=0$. We test strategies in which each available
lender quotes $p_L$ when $x_1=m$ and $p_H$ when $x_1<m$.
Both lenders observe $x_1$ under OA, so these strategies are
admissible. We verify each lender's best response while holding
its rival's proposed strategy fixed.

\noindent For lender 2, matching $p_H$ rather than undercutting gives a
payoff gain of $p_H(g_{x_1}(\beta)-\alpha)>0$ at a borrower
arrival whenever $1\le x_1<m$. At $(m,1)$, lender 1 instead
quotes $p_L$. Appendix~\ref{app:theorem1} shows that lender 2
prefers matching $p_L$ to quoting $p_H$ and letting lender 1
receive the borrower. Lender 2's quote $p_H$ when alone is
also optimal. These choices give its proposed strategy.

\noindent For lender 1, the quote at own available supply $m$ must be
the same in $(m,0)$ and $(m,1)$. Fix this common quote to
either $p_L$ or $p_H$. Under either choice,
Appendix~\ref{app:theorem1} verifies that lender 1 optimally
quotes $p_H$ whenever $1\le x_1<m$, both when lender 2 is
available and when it is not. It therefore remains to choose
between the two complete strategies that differ only in their
quote at $x_1=m$.

\noindent Starting from $s^0=(m,1)$, lender 1's payoff gain from choosing
$p_L$ rather than $p_H$ in both $(m,0)$ and $(m,1)$ is
\[
\frac{p_H\beta}{2}
\bigl[
(1+\beta)(\alpha-g_m(\beta))
+(1-\beta)(1-g_m(\beta))
\bigr]>0.
\]
Both terms in brackets are positive because
$\alpha>g_m(\beta)$ and $g_m(\beta)<g_0(\beta)=1$.
This establishes lender 1's proposed best response, including
the requirement that its quote cannot depend on lender 2's
availability. 

\noindent\textbf{4. Positive returns and equilibrium strategies.}
We now allow units to return, keeping $\lambda$ and $\delta$
as chosen in Step~1. Define
\[
\eta=
\min\bigl\{
\alpha-g_m(\beta),\,
g_{m-1}(\beta)-\alpha
\bigr\}>0.
\]
This is the smaller of lender 2's payoff gains from undercutting
at $x_1=m$ and matching at $x_1=m-1$ in Step~1, divided by
$p_H$. Appendix~\ref{app:theorem1} uses $\eta$ to bound the
payoff gains establishing the best responses in Steps~2 and~3,
then bounds how much introducing returns can change those gains.

\noindent Choose the return rate per outstanding unit to satisfy
\[
0<\frac{\mu}{\delta}
<
\frac{(1-\beta)\eta}
     {2(m+1)(1-\beta^m)}.
\]
\noindent The appendix shows that every required payoff gain remains
positive under this restriction. For lender 1 under OA, this
includes verifying its quotes with $x_1<m$ under either choice
of its common quote in $(m,0)$ and $(m,1)$, and then choosing
between those two complete strategies. All continuation payoffs
include the possibility of returning to $x_1=m$.

\noindent Consequently, the best responses from Steps~2 and~3 establish
the unique AO equilibrium, denoted $\sigma^{\mathrm{AO}}$,
and an OA equilibrium, denoted $\sigma^{\mathrm{OA}}$.
Their quotes are summarised below; each entry applies only
when the corresponding lender has an available unit.
\begin{center}
\renewcommand{\arraystretch}{1.15}
\begin{tabularx}{\linewidth}{@{}lXX@{}}
\toprule
Eq. & Lender 1 & Lender 2 \\
\midrule
$\sigma^{\mathrm{AO}}$
& $p_H$ when alone; $p_L$ otherwise
& $p_L$ whenever available \\
$\sigma^{\mathrm{OA}}$
& $p_H$ if $1\le x_1<m$; $p_L$ if $x_1=m$
& $p_H$ if $x_1<m$; $p_L$ if $x_1=m$ \\
\bottomrule
\end{tabularx}
\end{center}

\medskip
\noindent\textbf{5. Total lender payoff.}
The equilibria from Step~4 give the same allocation probabilities
in every state, since quotes tie whenever both lenders are
available. Return transitions are also identical. Their
transition matrices therefore coincide, giving the same expected
discounted time in each state from $s^0=(m,1)$. Using the
previously defined $z_s(\sigma)$, write $
z_{x_1,x_2}
:=
z_{(x_1,x_2)}(\sigma^{\mathrm{AO}})
=
z_{(x_1,x_2)}(\sigma^{\mathrm{OA}})$.

\noindent Relative to $\sigma^{\mathrm{AO}}$, the fee received per booking
under $\sigma^{\mathrm{OA}}$ is $p_H-p_L$ higher in states
$(x,1)$ with $0\le x<m$, and $p_H-p_L$ lower only in $(m,0)$.
Fees are unchanged elsewhere. It suffices to show that
$z_{m-1,1}>z_{m,0}$: the fee increase in $(m-1,1)$ then
outweighs the fee decrease in $(m,0)$, even before including
the increases in the other states.

\noindent In both $(m-1,1)$ and $(m,0)$, bookings occur at borrower
arrival rate $\lambda$, and the single outstanding unit
returns at rate $\mu$. The discounted balance equations are
\[
(\delta+\lambda+\mu)z_{m-1,1}=
\frac{\lambda}{2}z_{m,1}
+\mu z_{m-1,0}
+2\mu z_{m-2,1} \text{ and }
(\delta+\lambda+\mu)z_{m,0} =
\frac{\lambda}{2}z_{m,1}
+\mu z_{m-1,0}.
\]
The common terms record entry from $(m,1)$ through a booking
and from $(m-1,0)$ through a return. The additional term
$2\mu z_{m-2,1}$ arises because either of lender 1's two
outstanding units in $(m-2,1)$ can return, bringing the state
to $(m-1,1)$. There is no corresponding return into $(m,0)$,
since lender 2 owns only one unit.

Subtracting the equations gives
\[
z_{m-1,1}-z_{m,0}
=
\frac{2\mu}{\delta+\lambda+\mu}z_{m-2,1}>0.
\]
The inequality is strict because lender 1 can receive the first two borrowers before either unit returns, giving positive expected discounted time in $(m-2,1)$. 

\noindent The total payoff difference between $\mathrm{OA}$ and $\mathrm{AO}$ is bounded below by
\[
\pi^{OA}(\lambda)-\pi^{AO}(\lambda) \ge
\sum_{i=1}^{2}
\bigl[
V_i(s^0,\sigma^{\mathrm{OA}})
-
V_i(s^0,\sigma^{\mathrm{AO}})
\bigr] = \lambda(p_H-p_L)
\left[
\sum_{x=0}^{m-2}z_{x,1}
+\frac{2\mu}{\delta+\lambda+\mu}z_{m-2,1}
\right]>0.
\]

\medskip
\noindent\textbf{Nonempty parameter region.}
It remains to verify that the choices in Steps~1 and~4 are
possible for every supply and fee set in the theorem.
The equations $g_0(\beta)=1$ and  $g_j(\beta)=\frac{1}{2} +\frac{\beta}{2}g_{j-1}(\beta)$ show by induction that, for $1\le j\le m$, $g_j(\beta)$ is
continuous and strictly increasing in $\beta$, from
$\frac{1}{2}$ at $\beta=0$ to $1$ at $\beta=1$. For every $\alpha\in(\frac{1}{2},1)$, there is therefore a
value $\beta\in(0,1)$ satisfying $g_{m-1}(\beta)=\alpha$.
At that value, $g_m(\beta)<\alpha$. A sufficiently small
increase in $\beta$ gives
\[
g_m(\beta)<\alpha<g_{m-1}(\beta).
\]
Thus the choice required in Step~1 holds on a nonempty
interval of $\beta$. For any such $\beta$, choose any $\delta>0$ and set
\[
\lambda=\frac{\beta\delta}{1-\beta}.
\]
Then $\eta>0$, so the upper bound in Step~4 permits positive
values of $\mu$. All restrictions are strict and depend
continuously on $(\lambda,\mu,\delta)$. They therefore define
a nonempty open set of positive parameter triples for which
the stated equilibria and payoff inequality hold.
\end{proof}

\noindent For supplies $(2,1)$, direct equilibrium verification gives sufficient conditions for the same equilibrium quotes without restricting the ratio $\frac{\mu}{\delta}$.

\begin{proposition}[Two fees and supplies $(2,1)$] \label{prop:two-fee-region}
Consider supplies $(c_1,c_2)=(2,1)$, initial state $s^0=(2,1)$, and permitted fees $P=\{p_L,p_H\}$ with $0<p_L<p_H$. Suppose $\mu,\delta>0$ and
\[
0.74\le\frac{p_L}{p_H}\le0.76,
\qquad
\frac{13}{10}(\mu+\delta)\le\lambda
\le\frac{11}{8}(\mu+\delta).
\]
Then AO and OA each have a unique pure stationary equilibrium, and $\pi^{OA}(\lambda)>\pi^{AO}(\lambda)$.
\end{proposition}

\begin{proof}
Appendix~\ref{app:two-fee-region} computes the lenders' best responses, yielding the unique equilibria described in the proof of Theorem~\ref{thm:general-larger-supply} with $m=2$. Let $z_{0,1}$ denote their common expected discounted time in state
$(0,1)$. The payoff identity in that proof gives
\[
\pi^{OA}(\lambda)-\pi^{AO}(\lambda)
=\lambda(p_H-p_L)z_{0,1}
 \frac{\delta+\lambda+3\mu}{\delta+\lambda+\mu}>0,
\]
since $(0,1)$ is reachable from $(2,1)$.
\end{proof}

\noindent Proposition~\ref{prop:two-fee-region} establishes a sufficient
parameter region in which $OA$ yields higher total lender payoff than
$AO$. Figure~\ref{fig:main-regime-map}b illustrates this region for
$\delta=0.1$.

\subsubsection{Numerical comparison}
\label{par:numerical-comparison}

Motivated by the payoff reversal illustrated in Example~\ref{prop_supply_info}, we next examine the $AO$--$OA$ payoff difference numerically over a
broader range of borrower arrival and return rates. 
Figure~\ref{fig:main-regime-map}a presents the results for supplies
$(c_1,c_2)=(2,1)$, permitted fees $P=\{0.75,1\}$ and discount rate
$\delta=0.1$. We vary the borrower arrival rate $\lambda$ and return
rate per outstanding unit $\mu$ independently over
$\{0.1,0.2,\ldots,5.0\}$. At each parameter pair, we compute payoffs
from~\eqref{eq:payoff-system} and identify pure stationary equilibria
by checking every admissible unilateral deviation, accepting payoff
gains of at most $10^{-10}$.

\subsubsection{Pareto improvements}
\label{sec:pareto-improvements}
We now turn our attention to whether asymmetric information can
generate a Pareto improvement, meaning that both lenders obtain
strictly higher payoffs.

\begin{proposition}[Both lenders gain from informing one] \label{prop:pareto-gain}
There exists a lending game with unequal supplies $c_1>c_2$ in which informing only the smaller lender rather than only the larger yields a strict Pareto improvement for the lenders. The unique equilibria $\sigma^{AO}$ and $\sigma^{OA}$ satisfy
\[
V_i(s^0,\sigma^{OA})>V_i(s^0,\sigma^{AO}), \qquad i\in\{1,2\}.
\]
\end{proposition}

\begin{proof}
Choose supplies $(c_1,c_2)=(3,1)$, fee set $P=\{0.75,1\}$, arrival rate $\lambda=2$, return rate $\mu=1$ and discount rate $\delta=0.1$. The unique equilibria under $AO$ and $OA$ give the following payoffs.
\[
\begin{array}{c|rr}
 & V_i(s^0,\sigma^{AO}) & V_i(s^0,\sigma^{OA})\\ \hline 
\text{Lender 1} & 11.7404 & 11.9532 \\
\text{Lender 2} & 4.0717 & 5.4552
\end{array}.
\]
Both payoff differences are strictly positive. The equilibrium verification and exact payoff comparisons are given in Appendix~\ref{app:pareto}.
\end{proof}

\subsection{Ownership fragmentation} \label{sec:dividing-supply}

We split lender 2 in the $(2,2)$ OA game into two independent informed lenders with one unit each, giving $(2,1,1)$ under OAA. The following proposition shows that this split can raise total equilibrium payoff without increasing lending. Figure~\ref{fig:splitting-region} illustrates
the sufficient parameter region established by the
proposition.

\begin{proposition}[Higher total payoff after splitting a lender] \label{prop:split-informed-lender}
Let $P=\{p_L,p_H\}$, with $0<p_L<p_H$, and suppose
\[
0<\delta\le\mu,
\qquad
4(\mu+\delta)\le\lambda\le4(\mu+\delta)+\frac{\mu}{5},
\qquad
0.899\le\frac{p_L}{p_H}\le0.901.
\]
For supplies $(2,2)$ under OA and $(2,1,1)$ under OAA, the pure stationary equilibria yield
\[
\pi^{\mathrm{OAA}}(\lambda)>\pi^{\mathrm{OA}}(\lambda).
\]
\end{proposition}

\begin{proof}
The computer-assisted strategy verification in Appendix~\ref{app:fragmentation} gives unique equilibria $\sigma^{\mathrm{OA}}$ and $\sigma^{\mathrm{OAA}}$ with the below quotes.
\begin{center}
\renewcommand{\arraystretch}{1.15}
\begin{tabularx}{\linewidth}{@{}lXX@{}}
\toprule
Regime & Lender 1 & Informed lender(s) \\
\midrule
$(2,2)$, OA
& $p_L$ whenever available
& $p_H$ when alone; $p_L$ otherwise \\
$(2,1,1)$, OAA
& $p_H$ if $x_1=1$; $p_L$ if $x_1=2$
& Each quotes $p_H$ when alone; $p_L$ otherwise \\
\bottomrule
\end{tabularx}
\end{center}
By Lemma~\ref{lem:lending-invariance}, $B(\sigma)$ has the same value in both games; write $B$ for this common value. Define $z_{a,b}:=z_{(a,b)}(\sigma^{\mathrm{OA}})$ for $0\le a,b\le2$. These discounted state times refer only to the $(2,2)$ game. Its successful bookings pay $p_H$ in $(0,1)$ and $(0,2)$, and $p_L$ otherwise. Hence
\[
\pi_{\mathrm{OA}}(\lambda)
=
p_LB+\lambda(p_H-p_L)(z_{0,1}+z_{0,2}).
\]
After the split, a booking pays $p_H$ exactly when only one unit is available in total. If at least two units are available, either lender 1 has both its units available and quotes $p_L$, or an informed lender faces competing supply and quotes $p_L$. By Lemma~\ref{lem:lending-invariance} the discounted time with exactly one unit available is the same in both games. In the $(2,2)$ game this is $z_{1,0}+z_{0,1}$, so
\[
\pi^{\mathrm{OAA}}(\lambda) = p_LB+\lambda(p_H-p_L)(z_{1,0}+z_{0,1}).
\]
Subtracting gives
\[
\pi^{\mathrm{OAA}}(\lambda)-\pi^{\mathrm{OA}}(\lambda) =\lambda(p_H-p_L)(z_{1,0}-z_{0,2}).
\]
It remains to show that $z_{1,0}>z_{0,2}$. In the two-lender equilibrium, both lenders have supply two and receive equal allocation probabilities whenever both are available. Supply transitions are therefore symmetric between the lenders, despite their different information. Starting from $(2,2)$, this gives $z_{0,1}=z_{1,0}$, $z_{0,2}=z_{2,0}$ and $z_{1,2}=z_{2,1}$. Using these equalities, the discounted balance equations for $(0,2)$, $(1,1)$ and $(1,0)$, respectively, are
\[
\begin{aligned}
(\delta+\lambda+2\mu)z_{0,2}
&=\mu z_{1,0}+\frac{\lambda}{2}z_{1,2},\\
(\delta+\lambda+2\mu)z_{1,1}
&=4\mu z_{1,0}+\lambda z_{1,2},\\
(\delta+\lambda+3\mu)z_{1,0}
&=2\mu z_{0,0}+\lambda z_{0,2}
  +\frac{\lambda}{2}z_{1,1}.
\end{aligned}
\]
All these state times are positive because the states are reachable from $(2,2)$. Subtracting twice the first equation from the second gives $(\delta+\lambda+2\mu)(z_{1,1}-2z_{0,2}) =2\mu z_{1,0}>0$. Substitution into the third equation yields
\[
z_{1,0}>
\frac{2\lambda}{\delta+\lambda+3\mu}z_{0,2}
>z_{0,2},
\]
where the last inequality follows from $\lambda\ge4(\mu+\delta)>\delta+3\mu$. Since $\lambda>0$ and $p_H>p_L$, the payoff difference is positive.
\end{proof}
\noindent Splitting lender 2 removes the high fee it charged when it alone held two available units, since the two new lenders then compete. However, lender 1 changes its quote from $p_L$ to $p_H$ when it has one unit available, including when alone.

\section{Conclusion and Future Work}
We studied how access to aggregate rival-supply information affects lender payoffs, within observation-constrained pure stationary strategies. Informing only the larger lender can yield higher total payoff than informing only the smaller when demand is low relative to supply. With all other parameters unchanged, this ranking can reverse when demand is close to supply: informing only the smaller lender then yields the highest total payoff among all four information regimes. With two permitted fees, we established sufficient conditions for higher total payoff under OA than AO for any larger-lender supply when the smaller lender owns one unit. Reassigning information can also benefit both lenders, and splitting an informed lender can raise total payoff. At fixed total supply and arrival, return and discount rates, these payoff differences come from fees and discounted state times, not
additional lending.

\noindent Future work could allow loan end dates to be known at booking, so that strategies account for when units will return. A continuous
interval of permitted fees would allow arbitrarily small undercuts, requiring a fresh analysis of equilibrium existence and payoff rankings. Extending the comparisons to more lenders would also allow us to distinguish the effects of observing aggregate rival available supply from observing each rival's available supply separately.

\newpage
\bibliography{refs}
\bibliographystyle{plainnat}

\clearpage

\appendix
\section{Information Assignment: Proofs and Equilibrium Verification}
In this section we show all the proofs for Section \ref{sec_res_supply_info}. For each strategy profile $\sigma$, we construct $Q(\sigma)$ and $\mathbf r_i(\sigma)$, compute $\mathbf V_i(\sigma)$ from~\eqref{eq:payoff-system}, and compare $V_i(s^0,\sigma)$ across all admissible unilateral deviations. In each table, the first block reports a lender's best response to every remaining strategy of its rival. The second checks the rival's responses to every distinct best response in the first. An equilibrium must be a pair of mutual best responses appearing in these blocks. The column headings identify the state or own supply determining each quote. Both payoff columns refer to the responding lender. The next-best payoff is the highest payoff from any different remaining strategy against the same rival strategy. Comparisons use exact arithmetic, displayed payoffs are rounded to four decimals and bold entries identify equilibrium strategies and their payoffs.

\subsection{Example~\ref{prop_supply_info}: The recipient of supply information}
\subsubsection{Equilibrium verification for $\lambda=1$}
\label{app:lambda-one}
Fix supplies $(c_1,c_2)=(2,1)$, permitted fees $P=\{0.75,0.875,1\}$, $\lambda=\mu=1$, $\delta=0.1$, and initial state $s^0=(2,1)$. We establish the unique equilibria under $AO$ and $OA$.
\paragraph{Regime $AO$.}
Lender 2 chooses one quote whenever available. Lender 1 chooses
separate quotes in states $(1,1)$ and $(2,1)$. There are three lender-2 strategies and
nine remaining lender-1 strategies. The first block covers
every lender-2 quote; the second checks lender 2's responses
to every lender-1 strategy appearing in the first.

\begin{center}{%
\begin{tabular}{crrrr}
\hline
\shortstack{Lender 2's quote\\when $x_2=1$} &
\multicolumn{2}{c}{Lender 1's best response} &
Payoff &
\shortstack{Next-best\\payoff}\\
\cline{2-3}
 & $s=(1,1)$ & $s=(2,1)$ & & \\
\hline
\textbf{0.75}
 & \textbf{0.75} & \textbf{0.75}
 & \textbf{5.0688} & 4.7899\\
0.875
 & 0.75 & 0.75
 & 6.2918 & 5.9827\\
1
 & 0.875 & 0.875
 & 7.2385 & 6.8139\\
\hline
\shortstack{Lender 1's quotes\\in $(1,1)$ and $(2,1)$} &
\multicolumn{2}{c}{Lender 2's best response} &
Payoff &
\shortstack{Next-best\\payoff}\\
\cline{2-3}
 & \multicolumn{2}{c}{$x_2=1$} & & \\
\hline
$\boldsymbol{(0.75,0.75)}$
 & \multicolumn{2}{c}{\textbf{0.75}}
 & \textbf{2.7081} & 1.2711\\
$(0.875,0.875)$
 & \multicolumn{2}{c}{0.75}
 & 3.9286 & 3.1594\\
\hline
\end{tabular}}
\end{center}
Every displayed best response is unique. Lender 2 responds to either possible equilibrium strategy of lender 1 with 0.75. Lender 1 then responds with 0.75 in both competitive states. This is the only pair of mutual best responses: lender 1 quotes 1 when alone and 0.75 when competing, while lender 2 always quotes 0.75.

\paragraph{Regime $OA$.}
Lender 1 chooses separate quotes at own available supply 1 and 2. Lender 2 chooses separate quotes in states $(1,1)$ and $(2,1)$, with its quote when alone fixed at 1. Each lender has nine remaining strategies. The first block covers every lender-1 strategy; the second checks lender 1's responses to every lender-2 strategy appearing in the first.

\begin{center}
{%
\begin{tabular}{crrrr}
\hline
\shortstack{Lender 1's quotes\\at $x_1=1$ and $x_1=2$} &
\multicolumn{2}{c}{Lender 2's best response} &
Payoff &
\shortstack{Next-best\\payoff}\\
\cline{2-3}
 & $s=(1,1)$ & $s=(2,1)$ & & \\
\hline
$\boldsymbol{(0.75,0.75)}$
 & \textbf{0.75} & \textbf{0.75}
 & \textbf{2.8340} & 2.4583\\
$(0.75,0.875)$
 & 0.75 & 0.75
 & 3.8226 & 3.6455\\
$(0.75,1)$
 & 0.75 & 0.875
 & 4.3446 & 4.1676\\
$(0.875,0.75)$
 & 0.75 & 0.75
 & 3.1520 & 2.9614\\
$(0.875,0.875)$
 & 0.75 & 0.75
 & 3.9724 & 3.8826\\
$(0.875,1)$
 & 0.75 & 0.875
 & 4.4945 & 4.4046\\
$(1,0.75)$
 & 0.875 & 0.75
 & 3.3871 & 3.1520\\
$(1,0.875)$
 & 0.875 & 0.75
 & 4.0832 & 3.9724\\
$(1,1)$
 & 0.875 & 0.875
 & 4.6053 & 4.4945\\
\hline
\shortstack{Lender 2's quotes\\in $(1,1)$ and $(2,1)$} &
\multicolumn{2}{c}{Lender 1's best response} &
Payoff &
\shortstack{Next-best\\payoff}\\
\cline{2-3}
 & $x_1=1$ & $x_1=2$ & & \\
\hline
$\boldsymbol{(0.75,0.75)}$
 & \textbf{0.75} & \textbf{0.75}
 & \textbf{4.3841} & 4.3616\\
$(0.75,0.875)$
 & 0.75 & 0.75
 & 5.4161 & 5.1819\\
$(0.875,0.75)$
 & 0.75 & 0.75
 & 4.8569 & 4.6592\\
$(0.875,0.875)$
 & 0.75 & 0.75
 & 6.1388 & 5.7367\\
\hline
\end{tabular}}
\end{center}
Again, every displayed best response is unique. Lender 1 responds to all four possible equilibrium strategies of lender 2 with 0.75 at both own-supply levels. Lender 2 then responds with 0.75 in both competitive states. Hence the unique $OA$ equilibrium has lender 1 always quoting 0.75 and lender 2 quoting 1 when alone and 0.75 when competing.

\noindent Under $OO$, both lenders always quoting 0.75 is also an equilibrium. Under $AA$, both quoting 1 when alone and 0.75 when competing is an equilibrium. These follow from the own-supply and informed best responses above, respectively: changing a rival's quote when it is alone affects neither the lender's reward nor the supply transitions under any deviation.

\subsubsection{Equilibrium verification for $\lambda=3$}
\label{app:lambda-three}

Fix supplies $(c_1,c_2)=(2,1)$, permitted fees $P=\{0.75,0.875,1\}$, $\lambda=3$, $\mu=1$, $\delta=0.1$, and initial state $s^0=(2,1)$. We establish unique equilibria under $OO$, $AO$ and $OA$, and identify both equilibria under $AA$. The probability $q$ used in Appendix~\ref{app:example-payoffs}
describes the equilibrium allocations only. Deviations may produce different allocation probabilities in states $(1,1)$ and $(2,1)$.

\paragraph{Regime $OO$.}
Lender 1 chooses a quote at each own-supply level, $x_1=1$ and $x_1=2$. Lender 2 chooses one quote whenever $x_2=1$. There are nine lender-1 strategies and three lender-2 strategies.

\begin{center}{%
\begin{tabular}{crrrr}
\hline
\shortstack{Lender 2's quote\\when $x_2=1$} &
\multicolumn{2}{c}{Lender 1's best response} &
Payoff & \shortstack{Next-best\\payoff}\\
\cline{2-3}
 & $x_1=1$ & $x_1=2$ & & \\
\hline
0.75 & 1 & 1 & 12.3825 & 11.8361\\
\textbf{0.875} & \textbf{1} & \textbf{1} & \textbf{12.3825} & 12.2239\\
1 & 1 & 0.875 & 13.4266 & 13.4149\\
\hline
\shortstack{Lender 1 quotes\\$x_1=1,2$} &
\multicolumn{2}{c}{Lender 2's best response} &
Payoff & \shortstack{Next-best\\payoff}\\
\cline{2-3}
 & \multicolumn{2}{c}{$x_2=1$} & & \\
\hline
$(1,0.875)$ & \multicolumn{2}{c}{0.875} & 6.6549 & 6.3986\\
$\boldsymbol{(1,1)}$ & \multicolumn{2}{c}{\textbf{0.875}} & \textbf{7.0427} & 7.0164\\
\hline
\end{tabular}}
\end{center}
Every displayed best response is unique. Lender 2's response to either possible equilibrium strategy of lender 1 is 0.875. Lender 1 then responds with 1 at both own-supply levels. This gives the unique $OO$ equilibrium.

\paragraph{Regime $AO$.}
Lender 1 chooses competitive quotes in states $(1,1)$ and $(2,1)$, with its quotes when alone fixed at 1. Lender 2 chooses one quote whenever $x_2=1$. There are nine remaining lender-1 strategies and three lender-2 strategies.

\begin{center}
{%
\begin{tabular}{crrrr}
\hline
\shortstack{Lender 2's quote\\when $x_2=1$} &
\multicolumn{2}{c}{Lender 1's best response} &
Payoff & \shortstack{Next-best\\payoff}\\
\cline{2-3}
 & $s=(1,1)$ & $s=(2,1)$ & & \\
\hline
0.75 & 0.75 & 0.75 & 12.4529 & 12.4127\\
\textbf{0.875} & \textbf{0.875} & \textbf{0.875} & \textbf{12.9339} & 12.8510\\
1 & 0.875 & 0.875 & 13.7402 & 13.6198\\
\hline
\shortstack{Lender 1 quotes\\$(1,1),\ (2,1)$} &
\multicolumn{2}{c}{Lender 2's best response} &
Payoff & \shortstack{Next-best\\payoff}\\
\cline{2-3}
 & \multicolumn{2}{c}{$x_2=1$} & & \\
\hline
$(0.75,0.75)$ & \multicolumn{2}{c}{1} & 5.5067 & 5.2623\\
$\boldsymbol{(0.875,0.875)}$ & \multicolumn{2}{c}{\textbf{0.875}} & \textbf{6.1393} & 6.0366\\
\hline
\end{tabular}}
\end{center}
Lender 1's only possible equilibrium competitive quotes are $(0.75,0.75)$ and $(0.875,0.875)$. The former cannot be sustained: lender 2 responds with 1, to which lender 1's response is $(0.875,0.875)$. The unique pair of mutual best responses has lender 1 quote 0.875 in both competitive states and lender 2 quote 0.875 whenever available.

\paragraph{Regime $OA$.}
Lender 1 chooses quotes at own available supply $x_1=1$ and $x_1=2$. Lender 2 chooses competitive quotes in states $(1,1)$ and $(2,1)$, with its quote when alone fixed at 1. Each lender has nine remaining strategies.

\begin{center}
\adjustbox{max width=\linewidth}{%
\begin{tabular}{crrrr}
\hline
\shortstack{Lender 1 quotes\\$x_1=1,2$} &
\multicolumn{2}{c}{Lender 2's best response} &
Payoff & \shortstack{Next-best\\payoff}\\
\cline{2-3}
 & $s=(1,1)$ & $s=(2,1)$ & & \\
\hline
$(0.75,0.75)$ & 0.75 & 0.75 & 6.0543 & 5.9154\\
$(0.75,0.875)$ & 0.75 & 0.75 & 6.4219 & 6.3242\\
$(0.75,1)$ & 0.75 & 0.875 & 6.8349 & 6.7371\\
$(0.875,0.75)$ & 0.875 & 0.75 & 6.3289 & 6.1733\\
$(0.875,0.875)$ & 0.875 & 0.75 & 6.6151 & 6.5353\\
$(0.875,1)$ & 0.875 & 0.875 & 7.0280 & 6.9186\\
$(1,0.75)$ & 0.875 & 0.75 & 6.6830 & 6.6034\\
$(1,0.875)$ & 0.875 & 0.875 & 6.8895 & 6.8643\\
$\boldsymbol{(1,1)}$ & \textbf{0.875} & \textbf{0.875} & \textbf{7.2772} & 7.2212\\
\hline
\shortstack{Lender 2 quotes\\$(1,1),\ (2,1)$} &
\multicolumn{2}{c}{Lender 1's best response} &
Payoff & \shortstack{Next-best\\payoff}\\
\cline{2-3}
 & $x_1=1$ & $x_1=2$ & & \\
\hline
$(0.75,0.75)$ & 1 & 1 & 12.3825 & 11.8361\\
$(0.75,0.875)$ & 1 & 1 & 12.3825 & 12.2239\\
$(0.875,0.75)$ & 1 & 1 & 12.3825 & 11.8361\\
$\boldsymbol{(0.875,0.875)}$ & \textbf{1} & \textbf{1} & \textbf{12.3825} & 12.2239\\
\hline
\end{tabular}}
\end{center}
Lender 1 responds to every possible equilibrium strategy of lender 2 with 1 at both own-supply levels. Lender 2 then responds with 0.875 in both competitive states. Together with its quote of 1 when alone, this gives the unique $OA$ equilibrium.

\paragraph{Regime $AA$.}
Both lenders quote 1 when alone. Each chooses a pair of competitive quotes for states $(1,1)$ and $(2,1)$, giving nine remaining strategies per lender.

\begin{center}{%
\begin{tabular}{crrrr}
\hline
\shortstack{Lender 1 quotes\\$(1,1),\ (2,1)$} &
\multicolumn{2}{c}{Lender 2's best response} &
Payoff & \shortstack{Next-best\\payoff}\\
\cline{2-3}
 & $s=(1,1)$ & $s=(2,1)$ & & \\
\hline
$\boldsymbol{(0.75,0.75)}$ & \textbf{0.75} & \textbf{0.75} & \textbf{6.0543} & 5.9154\\
$(0.75,0.875)$ & 0.75 & 0.75 & 6.4219 & 6.3242\\
$(0.75,1)$ & 0.75 & 0.875 & 6.8349 & 6.7371\\
$\boldsymbol{(0.875,0.75)}$ & \textbf{0.875} & \textbf{0.75} & \textbf{6.3289} & 6.1733\\
$(0.875,0.875)$ & 0.875 & 0.75 & 6.6151 & 6.5353\\
$(0.875,1)$ & 0.875 & 0.875 & 7.0280 & 6.9186\\
$(1,0.75)$ & 0.875 & 0.75 & 6.6830 & 6.6034\\
$(1,0.875)$ & 0.875 & 0.875 & 6.8895 & 6.8643\\
$(1,1)$ & 0.875 & 0.875 & 7.2772 & 7.2212\\
\hline
\shortstack{Lender 2 quotes\\$(1,1),\ (2,1)$} &
\multicolumn{2}{c}{Lender 1's best response} &
Payoff & \shortstack{Next-best\\payoff}\\
\cline{2-3}
 & $s=(1,1)$ & $s=(2,1)$ & & \\
\hline
$\boldsymbol{(0.75,0.75)}$ & \textbf{0.75} & \textbf{0.75} & \textbf{12.4529} & 12.4127\\
$(0.75,0.875)$ & 0.75 & 0.875 & 12.6594 & 12.6192\\
$\boldsymbol{(0.875,0.75)}$ & \textbf{0.875} & \textbf{0.75} & \textbf{12.7275} & 12.6040\\
$(0.875,0.875)$ & 0.875 & 0.875 & 12.9339 & 12.8510\\
\hline
\end{tabular}}
\end{center}
The first block restricts lender 2's equilibrium competitive quotes to the four pairs checked in the second. Lender 1's unique response to each is to match that pair. Returning to the first block, matching is also lender 2's best response only for $(0.75,0.75)$ and $(0.875,0.75)$. There are therefore exactly two $AA$ equilibria. In the first, both lenders quote 0.75 in both competitive states. In the second, both quote 0.875 in $(1,1)$ and 0.75 in $(2,1)$. Both equilibria have quotes of 1 when alone. Their total payoffs are 18.5072 and 19.0564, respectively, so $\pi^{\mathrm{AA}}(3)=19.0564$, rounded to four decimal places.

\subsubsection{Payoff comparisons}
\label{app:example-payoffs}

\begin{proof}[Proof of Example~\ref{prop_supply_info}]
We begin with the proof for the lower arrival rate. Choose supplies
$(c_1,c_2)=(2,1)$, permitted fees $P=\{0.75,0.875,1\}$, return rate
$\mu=1$, discount rate $\delta=0.1$, initial state $s^0=(2,1)$ and
$\lambda_L=1$.

\noindent Appendix~\ref{app:lambda-one} establishes the unique
equilibria under $AO$ and $OA$. In each regime, both lenders quote
$\frac{3}{4}$ when competing. When alone, the informed lender
quotes 1 and the uninformed lender quotes $\frac{3}{4}$. When both
lenders have available supply, their quotes tie, so each receives
the borrower with probability $\frac{1}{2}$ under either regime.
When only one lender has supply, it receives the borrower with
probability 1, since all permitted fees are accepted. No booking
occurs when neither has supply. Under these equilibrium strategies,
each lender's probability of receiving a borrower in a given supply
state is the same under $AO$ and $OA$. The regimes differ only in
the fee charged when a lender is alone. Return rates are also
unchanged, giving the transition-rate matrix
\[
Q=
\begin{array}{c|rrrrrr}
 & (0,0)&(0,1)&(1,0)&(1,1)&(2,0)&(2,1)\\
\hline
(0,0)&-3&1&2&0&0&0\\
(0,1)&1&-3&0&2&0&0\\
(1,0)&1&0&-3&1&1&0\\
(1,1)&0&\frac{1}{2}&\frac{1}{2}&-2&0&1\\
(2,0)&0&0&1&0&-2&1\\
(2,1)&0&0&0&\frac{1}{2}&\frac{1}{2}&-1
\end{array},
\]
with rows indexing the current state and columns the destination.

\noindent Let $\sigma^{AO}$ and $\sigma^{OA}$ denote the two
equilibrium strategy profiles. Since they induce the same
transition-rate matrix and start from the same state, their
discounted times in every state are the same. Using the previously
defined $z_s(\sigma)$, write
\[
z_{ab}
:=
z_{(a,b)}(\sigma^{AO})
=
z_{(a,b)}(\sigma^{OA}),
\qquad (a,b)\in S.
\]
Applying~\eqref{eq:state-time-system} to the displayed matrix $Q$,
the rows corresponding to $(1,0)$ and $(0,1)$ give
\[
(\delta+3)z_{10}
=2z_{00}+\frac{1}{2}z_{11}+z_{20}
\text{ and }
(\delta+3)z_{01}
=z_{00}+\frac{1}{2}z_{11}.
\]
Subtracting yields
\[
z_{10}-z_{01}
=
\frac{z_{00}+z_{20}}{\delta+3}>0,
\]
since both $(0,0)$ and $(2,0)$ are reachable from $(2,1)$ and have
positive expected holding times. Relative to $OA$, the fee
received in regime $AO$ is $\frac{1}{4}$ more per booking in states
$(1,0)$ and $(2,0)$, and $\frac{1}{4}$ less in $(0,1)$. Fees
coincide when both lenders are available. With $\lambda=1$,
\[
\pi^{AO}(1)-\pi^{OA}(1)
=
\frac{1}{4}\bigl(z_{10}+z_{20}-z_{01}\bigr)>0.
\]

We next give the proof for the higher arrival rate.
Set $\lambda_H=3$, retaining supplies $(c_1,c_2)=(2,1)$, permitted
fees $P=\{0.75,0.875,1\}$, return rate $\mu=1$ and discount rate
$\delta=0.1$. Appendix~\ref{app:lambda-three} establishes unique
equilibria under $OO$, $AO$ and $OA$, with fees
\begin{center}
\small
\setlength{\tabcolsep}{4pt}
\renewcommand{\arraystretch}{1.05}
\begin{tabularx}{\linewidth}{@{}cXX@{}}
\toprule
Regime & Lender 1 & Lender 2 \\
\midrule
$\mathrm{OO}$
    & $1$ whenever available
    & $0.875$ whenever available \\

$\mathrm{AO}$
    & $1$ when alone; $0.875$ otherwise
    & $0.875$ whenever available \\

$\mathrm{OA}$
    & $1$ whenever available
    & $1$ when alone; $0.875$ otherwise \\
\bottomrule
\end{tabularx}
\end{center}
Under $AA$, there are two equilibria. Both lenders quote 1 when
alone. When competing, they both quote 0.75 in the first
equilibrium; in the second, they both quote 0.875 in state $(1,1)$
and 0.75 in $(2,1)$. Under $OO$ and $OA$, lender 2 quotes less than
lender 1 whenever both are available, and therefore receives the
borrower with probability 1. Under $AO$ and both $AA$ equilibria,
their competitive fees tie, so each lender receives the borrower
with probability $\frac{1}{2}$.

\noindent Let $q$ denote lender 1's allocation probability when both
lenders are available. The equilibrium transition matrices are
therefore
\[
Q_q=
\begin{array}{c|rrrrrr}
 & (0,0)&(0,1)&(1,0)&(1,1)&(2,0)&(2,1)\\
\hline
(0,0)&-3&1&2&0&0&0\\
(0,1)&3&-5&0&2&0&0\\
(1,0)&3&0&-5&1&1&0\\
(1,1)&0&3q&3(1-q)&-4&0&1\\
(2,0)&0&0&3&0&-4&1\\
(2,1)&0&0&0&3q&3(1-q)&-3
\end{array},
\]
where $q=0$ under $OO$ and $OA$, and $q=\frac{1}{2}$ under $AO$
and both $AA$ equilibria. Rows index the current state and columns
the destination. For $R\in\{AO,OA\}$, let $\sigma^R$ denote the
equilibrium under regime $R$, and write
$z_{ab}^{R}:=z_{(a,b)}(\sigma^R)$. Thus $z_{ab}^{R}$ is the expected
discounted time spent in state $(a,b)$ under that equilibrium.
Because their transition matrices coincide, $OO$ has the same
discounted state times as $OA$, and both $AA$ equilibria have the
same discounted state times as $AO$.
Lemma~\ref{lem:lending-invariance} gives equal discounted time at
each total available supply. Total available supply zero corresponds
only to $(0,0)$, total available supply three only to $(2,1)$, and
total available supply one to $(0,1)$ or $(1,0)$. Hence
\[
z_{00}^{AO}=z_{00}^{OA},\
z_{21}^{AO}=z_{21}^{OA},
\text{ and }
z_{01}^{AO}+z_{10}^{AO}=z_{01}^{OA}+z_{10}^{OA}.
\]
Applying~\eqref{eq:state-time-system} to $Q_{\frac{1}{2}}$ under
$AO$ and $Q_0$ under $OA$, the row corresponding to $(0,1)$ gives
\[
(\delta+5)z_{01}^{AO}
=z_{00}^{AO}+\frac{3}{2}z_{11}^{AO}
\text{ and }
(\delta+5)z_{01}^{OA}=z_{00}^{OA}.
\]
Subtracting and using $z_{00}^{AO}=z_{00}^{OA}$ gives
$(\delta+5)(z_{01}^{AO}-z_{01}^{OA})
=\frac{3}{2}z_{11}^{AO}>0$.

\noindent The row corresponding to $(1,1)$ gives
\[
(\delta+4)z_{11}^{AO}
=2z_{01}^{AO}+z_{10}^{AO}+\frac{3}{2}z_{21}^{AO}
\text{ and }
(\delta+4)z_{11}^{OA}=2z_{01}^{OA}+z_{10}^{OA}.
\]
With $z_{10}^{AO}-z_{10}^{OA}=-(z_{01}^{AO}-z_{01}^{OA})$,
subtracting yields
\[
(\delta+4)(z_{11}^{AO}-z_{11}^{OA})
=
z_{01}^{AO}-z_{01}^{OA}+\frac{3}{2}z_{21}^{AO}>0.
\]
The state times are positive because the states are reachable.
Thus $OA$ spends less discounted time in $(1,1)$ than $AO$, while
time in $(2,1)$ is unchanged. By
Lemma~\ref{lem:lending-invariance}, $B(\sigma)$ has the same value
under all four information regimes; write $B$ for this common value.
Under $OA$, each booking pays 1 unless both lenders are available,
in which case it pays 0.875. With arrival rate $\lambda=3$, the
reduction from a fee of 1 contributes $\frac{3}{8}$ per unit of
discounted time in either competitive state. Therefore,
\[
\pi^{OA}(3)
=
B-\frac{3}{8}\bigl(z_{11}^{OA}+z_{21}^{OA}\bigr).
\]
Under $OO$, allocations are unchanged, but lender 2 also charges
0.875 when alone. Hence
\[
\pi^{OA}(3)-\pi^{OO}(3)
=
\frac{3}{8}z_{01}^{OA}>0.
\]
Under $AO$, the fee is 0.875 when lenders compete and when only
lender 2 is available, so
\[
\pi^{AO}(3)
=
B-\frac{3}{8}\bigl(z_{01}^{AO}+z_{11}^{AO}+z_{21}^{AO}\bigr).
\]
Using $z_{21}^{AO}=z_{21}^{OA}$ gives
\[
\pi^{OA}(3)-\pi^{AO}(3)
=
\frac{3}{8}\bigl(z_{01}^{AO}+z_{11}^{AO}-z_{11}^{OA}\bigr)>0.
\]
Finally, both $AA$ equilibria have the discounted state times of
$AO$, charge 1 when alone, and at most 0.875 when competing. Their
payoffs therefore satisfy
\[
\pi^{AA}(3)
\leq
B-\frac{3}{8}\bigl(z_{11}^{AO}+z_{21}^{AO}\bigr).
\]
Consequently,
\[
\pi^{OA}(3)-\pi^{AA}(3)
\geq
\frac{3}{8}\bigl(z_{11}^{AO}-z_{11}^{OA}\bigr)>0.
\]
Together, these inequalities establish
\[
\pi^{OA}(3)>
\max\{\pi^{OO}(3),\pi^{AO}(3),\pi^{AA}(3)\}.\qedhere
\]
\end{proof}

\subsection{Rate bounds and equilibrium verification for Theorem~\ref{thm:general-larger-supply}} \label{app:theorem1}
Throughout this appendix, divide fees and payoffs by $p_H$, so the permitted fees are $\{\alpha,1\}$. Retain $r_i$ and $V_i$ for the resulting revenue and payoff vectors, and write $g_j=g_j(\beta)$.

\medskip
\noindent\textbf{1. Pricing incentives when units do not return.}
Set $\mu=0$. Iterating
$g_j=\frac{1}{2}+\frac{\beta}{2}g_{j-1}$ from $g_0=1$ gives
\[
g_j=
\frac{1+(1-\beta)\left(\frac{\beta}{2}\right)^j}{2-\beta},
\qquad j=0,\ldots,m.
\]
Consequently,
\[
g_{j-1}-g_j
=
\frac{1-\beta}{2}
\left(\frac{\beta}{2}\right)^{j-1}>0,
\qquad j=1,\ldots,m,
\]
which proves the decrease used in the main proof. To verify lender 1's response to lender 2's constant low quote in Step~2, we will need $\alpha-\beta>0$. This follows from
our choice $\alpha>g_m$, since
\[
g_m-\beta = \frac{1+(1-\beta)\left(\frac{\beta}{2}\right)^m -\beta(2-\beta)} {2-\beta} = \frac{(1-\beta) \left[1-\beta+\left(\frac{\beta}{2}\right)^m\right]} {2-\beta}>0.
\]
Hence $\alpha>g_m>\beta$.

\paragraph{Lender 1's incentive to match the high fee.}
Define
\[
F_0=0,
\qquad
F_x=\sum_{k=1}^{x}\beta^k,
\qquad x=1,\ldots,m+1.
\]
Here $F_x$ is the expected discounted sum of fees when each of the next $x$ borrowers pays $1$. The first payment carries the factor $\beta$ because that borrower has not yet arrived. The same initial wait must be included when using $g_j$ in a state payoff. The definition of $g_j$ starts with a borrower already present. Starting instead from state $(j,1)$ before the
next arrival, lender 2 must first wait for that borrower. Available supplies remain unchanged during this wait because $\mu=0$. Thus, when both lenders continue quoting $1$, lender 2's payoff from that state is $\beta g_j$. The extra factor discounts only the initial wait; $g_j$ already discounts all subsequent waits if lender 2 does not receive the first borrower. From state $(j,1)$, the next $j+1$ borrowers each pay $1$, giving total lender payoff $F_{j+1}$. Lender 1's payoff is therefore $F_{j+1}-\beta g_j$.
Now consider lender 1's quote at a borrower arrival in $(j,1)$, with $1\le j\le m$. Fix lender 2's current quote at $1$ and both lenders' subsequent quotes at $1$. If lender 1 receives the borrower, its continuation payoff from $(j-1,1)$ is $F_j-\beta g_{j-1}$. If lender 2 receives the borrower, lender 1 retains $j$ available units without competing supply, giving continuation payoff $F_j$. Both continuation payoffs include waiting for the next borrower after the current booking. Lender 1's payoff gain from matching $1$ rather than undercutting at $\alpha$ is therefore
\[
\frac{1}{2}\bigl(1+F_j-\beta g_{j-1}\bigr) +\frac{1}{2}F_j - \bigl(\alpha+F_j-\beta g_{j-1}\bigr) =\frac{1+\beta g_{j-1}}{2}-\alpha =g_j-\alpha.
\]
Thus the same expression gives either lender's payoff gain from
matching, under the specified subsequent quotes.

\paragraph{Verifying complete strategy changes.}
To use these quote incentives in the equilibrium verification,
we need to establish optimality against changes to complete
admissible strategies. The following identity applies at any
fixed $\mu\ge0$.

Fix both lenders' strategies $\sigma$, and let
$\widetilde{\sigma}$ differ only in lender $i$'s strategy.
Subtracting~\eqref{eq:payoff-system} for $\sigma$ and $\widetilde{\sigma}$ gives
\[
\begin{aligned}
\bigl[\delta\operatorname{Id}-Q(\widetilde{\sigma})\bigr]
\bigl[V_i(\sigma)-V_i(\widetilde{\sigma})\bigr]
={}&r_i(\sigma)-r_i(\widetilde{\sigma})\\
&+\bigl[Q(\sigma)-Q(\widetilde{\sigma})\bigr]V_i(\sigma).
\end{aligned}
\]
In each state, the right-hand side divided by the borrower
arrival rate $\lambda$ is lender $i$'s payoff gain from its
quote under $\sigma$ rather than its quote under
$\widetilde{\sigma}$ at the current arrival, with both lenders
following $\sigma$ thereafter. The return transitions cancel
because changing quotes does not change them.

The inverse matrix on the left satisfies
\[
\bigl[\delta\operatorname{Id}-Q(\widetilde{\sigma})\bigr]^{-1}
=
\int_0^\infty
e^{-\delta t}e^{tQ(\widetilde{\sigma})}\,dt.
\]
It has nonnegative entries and row sums $\frac{1}{\delta}$.
Hence, if the right-hand side is nonnegative in every state,
then
\[
V_i(s^0,\sigma)\ge V_i(s^0,\widetilde{\sigma}).
\]
Checking that the proposed quote is at least as profitable as
every alternative quote in each state, with both lenders
following $\sigma$ thereafter, is therefore sufficient to exclude
profitable changes to any or all of lender $i$'s admissible quotes.

\medskip
\noindent\textbf{2. Best responses under AO.}
Keep $\mu=0$. Lender 2's two admissible strategies are to quote
$\alpha$ whenever available or to quote $1$ whenever available.
For either choice, lender 1 optimally quotes $1$ in states
$(x,0)$, with $1\le x\le m$: raising its quote from $\alpha$
to $1$ increases its booking fee without changing any supply
transition. We verify its remaining quotes against each
lender-2 strategy.

\paragraph{Lender 1 against lender 2's constant quote $\alpha$.}
Consider lender 1's strategy that quotes $\alpha$ when both
lenders have available units and $1$ when alone. Let $U_x$
denote lender 1's payoff from state $(x,1)$ under these
strategies. At the next arrival, either lender 1 receives
$\alpha$ and continues from $(x-1,1)$, or lender 2 receives
the borrower and leaves lender 1 with $x$ available units
without competition. Thus
\[
U_0=0,
\qquad
U_x=\frac{\beta}{2}
\bigl(\alpha+U_{x-1}+F_x\bigr),
\qquad x=1,\ldots,m.
\]

At a borrower arrival in $(x,1)$, quoting $1$ instead of
matching $\alpha$ gives the borrower to lender 2 and leaves
lender 1 with continuation payoff $F_x$. Using
$F_x=\beta+\beta F_{x-1}$, induction in the equation for $U_x$
gives lender 1's payoff gain from matching, with its proposed
strategy followed thereafter:
\[
\frac{1}{2}\bigl(\alpha+U_{x-1}+F_x\bigr)-F_x = \frac{U_x-\beta F_x}{\beta} = \frac{\alpha-\beta}{\beta} \sum_{k=1}^{x}\left(\frac{\beta}{2}\right)^k \ge\frac{\alpha-\beta}{2}>0.
\]
The final inequality uses $\alpha>\beta$, established in
Step~1. Matching $\alpha$ therefore increases lender 1's
payoff at every state where both lenders have available units.

\paragraph{Lender 1 against lender 2's constant quote $1$.}
For $1\le x<m$, Step~1 gives lender 1 a payoff gain
$g_x-\alpha>0$ from matching $1$ rather than undercutting,
with both lenders quoting $1$ subsequently. These subsequent
quotes are optimal by induction on $x$: after the current
booking, lender 1 either has fewer available units or faces
no further competition.

At $(m,1)$, either allocation leads to a state where lender
1's verified quote is $1$: $(m-1,1)$ if lender 1 receives
the borrower, or $(m,0)$ if lender 2 receives it. Lender 1's
gain from undercutting rather than matching at $(m,1)$ is
therefore
\[
\alpha-g_m>0.
\]
Its best response consequently quotes $1$ when competing
with $x_1<m$, $\alpha$ when competing with $x_1=m$, and
$1$ when alone. The complete-strategy identity in Step~1
verifies both lender-1 responses against all admissible
strategy changes.

\paragraph{Lender 2's choice against each lender-1 response.}
Hold each of the two verified lender-1 strategies fixed.
Lender 2's payoffs from $s^0=(m,1)$ under its two constant
quotes are
\begin{center}
\renewcommand{\arraystretch}{1.3}
\begin{tabularx}{\linewidth}{@{}Xcc@{}}
\toprule
Lender 1's quotes when both lenders are available
& \shortstack{Lender 2 always\\quotes $\alpha$}
& \shortstack{Lender 2 always\\quotes $1$} \\
\midrule
$\alpha$ for $1\le x_1\le m$
& $\alpha\beta g_m$
& $\beta^{m+1}$ \\
$1$ for $1\le x_1<m$; $\alpha$ for $x_1=m$
& $\frac{\alpha\beta(1+\beta)}{2}$
& $\beta^2g_{m-1}$ \\
\bottomrule
\end{tabularx}
\end{center}

In the first row, matching $\alpha$ gives lender 2 the same
probabilities of receiving successive borrowers as in the
calculation defining $g_m$, but with fee $\alpha$ and an
initial wait for the first borrower. Quoting $1$ instead
postpones its booking until lender 1 has received $m$
borrowers, so lender 2 receives the $(m+1)$-st borrower.

In the second row, quoting $\alpha$ gives lender 2 either
the first borrower or, if lender 1 receives that borrower,
the second borrower; each possibility has probability
$\frac{1}{2}$. Quoting $1$ instead gives lender 1 the first
borrower and leaves lender 2 with continuation payoff
$\beta g_{m-1}$ from $(m-1,1)$. Discounting the wait for
the first borrower gives $\beta(\beta g_{m-1})$.

To establish that both rows favour $\alpha$, use the bounds
from Step~1. Since $m\ge2$ and $0<\beta<1$,
\[
g_m>\beta\ge\beta^{m-1},
\qquad
g_m\ge\frac{1}{2},
\qquad
g_{m-1}\le g_1=\frac{1+\beta}{2}.
\]
Lender 2's payoff gains from choosing $\alpha$ rather than
$1$ therefore satisfy
\[
\begin{aligned}
\alpha\beta g_m-\beta^{m+1}
&\ge \beta(\alpha-\beta)g_m
 \ge \frac{\beta}{2}(\alpha-\beta)>0,\\
\frac{\alpha\beta(1+\beta)}{2}-\beta^2g_{m-1}
&\ge \frac{\beta(1+\beta)}{2}(\alpha-\beta)>0.
\end{aligned}
\]
Thus lender 2 strictly prefers its constant quote $\alpha$
against either verified lender-1 response.

\medskip
\noindent\textbf{3. Best responses under OA.}
Keep $\mu=0$.

\paragraph{Lender 2's best response.}
Fix lender 1's proposed strategy: quote $\alpha$ when $x_1=m$
and $1$ when $1\le x_1<m$. For $1\le x_1<m$, lender 2's
payoff gain at a borrower arrival from matching $1$ rather
than undercutting is
\[
g_{x_1}-\alpha>0.
\]
These inequalities verify lender 2's high quotes at every
smaller positive rival available supply. At $(0,1)$, quoting
$1$ rather than $\alpha$ gives a payoff gain of $1-\alpha>0$.

At $(m,1)$, lender 1 quotes $\alpha$. If lender 2 also quotes
$\alpha$, it either receives that fee immediately or continues
from $(m-1,1)$ after lender 1 receives the borrower. In the
latter case, its continuation payoff is $\beta g_{m-1}$.
Quoting $1$ instead gives the current borrower to lender 1
with certainty. Lender 2's gain from matching $\alpha$ is thus
\[\frac{1}{2}\bigl(\alpha+\beta g_{m-1}\bigr) -\beta g_{m-1} = \frac{1}{2}\bigl(\alpha-\beta g_{m-1}\bigr) = \frac{1}{2} \bigl[(\alpha-g_m)+(1-g_m)\bigr]>0,
\]
where $\beta g_{m-1}=2g_m-1$ follows from the defining equation for $g_m$.

\paragraph{Lender 1's best response.}
Now hold lender 2's proposed strategy fixed: quote $\alpha$
in $(m,1)$ and $1$ in $(x,1)$ for $0\le x<m$.
Lender 1 must use the same quote in $(m,0)$ and $(m,1)$.
Fix that quote to either $\alpha$ or $1$.

Starting with $x_1<m$, available supply cannot return to $m$,
so continuation payoffs do not depend on this fixed quote.
At every state $(x,1)$ with $1\le x<m$, lender 1's gain from
matching $1$ rather than undercutting, with both lenders
subsequently quoting $1$, is $g_x-\alpha>0$, as derived
previously. At $(x,0)$, quoting $1$ rather than $\alpha$ gives a
gain of $1-\alpha>0$. Holding the quote at $x_1=m$ fixed,
the complete-strategy identity therefore excludes every
profitable change to lender 1's quotes with $x_1<m$.
Thus, under either fixed choice, lender 1 optimally quotes
$1$ at every smaller positive own available supply.

It remains to choose between the two resulting complete
strategies. If lender 1 quotes $1$ at $x_1=m$, lender 2
receives the first borrower at $\alpha$. Lender 1 then
continues alone from $(m,0)$, quoting $1$ throughout.
Its payoff from $s^0=(m,1)$ is therefore $\beta F_m$. If lender 1 instead quotes $\alpha$ at $x_1=m$, its payoff from $(m,0)$ is
\[
\beta\bigl(\alpha+F_{m-1}\bigr).
\]
The first booking there pays $\alpha$, because lender 1
cannot distinguish $(m,0)$ from $(m,1)$; its remaining
$m-1$ bookings pay $1$.

Starting from $(m,1)$, both lenders now quote $\alpha$.
If lender 1 receives the first borrower, it earns $\alpha$
and continues from $(m-1,1)$ with payoff
$F_m-\beta g_{m-1}$. If lender 2 receives that borrower,
lender 1 continues from $(m,0)$ with the payoff just
calculated. Including the wait for the first borrower,
lender 1's payoff is
\[
\frac{\beta}{2}
\bigl[
\alpha+F_m-\beta g_{m-1}
+\beta(\alpha+F_{m-1})
\bigr].
\]
Subtracting the payoff $\beta F_m$ from choosing $1$ at $x_1=m$, and using $F_m=\beta+\beta F_{m-1}$, gives
\[
\frac{\beta}{2}
\bigl[(1+\beta)\alpha-\beta-\beta g_{m-1}\bigr]=
\frac{\beta}{2}
\bigl[
(1+\beta)(\alpha-g_m)
+(1-\beta)(1-g_m)
\bigr]>0,
\]
where the equality uses $\beta g_{m-1}=2g_m-1$.
The inequality follows from $\alpha>g_m$ and $g_m<1$.
Lender 1 therefore prefers $\alpha$ at $x_1=m$, together
with $1$ at every smaller positive own available supply.
Since the preceding argument covers both possible quotes
at $x_1=m$ and all admissible choices below $m$, this
verifies its proposed complete strategy.

\medskip
\noindent\textbf{4. Positive returns and equilibrium strategies.}
We first bound the payoff gains from Steps~2 and~3 using
$\eta$, then bound their changes when $\mu>0$.

\paragraph{Payoff gains when $\mu=0$.}
The definition of $\eta$ and the decrease of $g_j$ with $j$
give
\[
g_j-\alpha\ge\eta\quad(1\le j<m),
\qquad
\alpha-g_m\ge\eta.
\]
Since $g_{m-1}<1$, we also have
\[
1-\alpha>\eta,
\qquad
1-g_m>\eta.
\]
To bound the gains involving $\alpha-\beta$, note that
$g_{m-1}-g_m\ge2\eta$. The defining equation for $g_m$ gives
\[ \alpha-\beta > g_m-\beta = (g_{m-1}-g_m )+(1-\beta)(1-g_{m-1}) >2\eta.
\]
Substituting these bounds into the expressions in Steps~2
and~3 shows that every verified gain from a proposed quote
rather than its alternative at a borrower arrival is at
least $\eta$. The two gains from lender 2 choosing its
constant quote $\alpha$ rather than $1$ under AO each exceed
$\beta\eta$. Lender 1's gain from choosing $\alpha$ rather
than $1$ at $x_1=m$ under OA also exceeds $\beta\eta$.

\paragraph{How returns change a lender's payoff.}
Fix any pair of admissible lender strategies and hold
$\lambda$ and $\delta$ unchanged. Write $Q^\mu$ for their
transition matrix and $V_i^\mu$ for lender $i$'s payoff
vector when units return at rate $\mu$ per outstanding unit.
Superscript $0$ denotes the same strategies with $\mu=0$.

Without returns, either lender can make at most $m$ bookings.
Its $k$-th booking cannot precede the $k$-th borrower arrival,
and each fee is at most $1$. Hence
\[
0\le V_i^0(s)\le F_m
\qquad\text{for every }s\in S.
\]
The booking revenue vector is unchanged when only $\mu$
changes. Subtracting~\eqref{eq:payoff-system} at return rates $\mu$ and $0$ gives
\[
(\delta\operatorname{Id}-Q^\mu)
(V_i^\mu-V_i^0)
=
(Q^\mu-Q^0)V_i^0.
\]
At state $s=(x_1,x_2)$, the right-hand side equals
\[
\bigl[(Q^\mu-Q^0)V_i^0\bigr](s)
=
\mu\sum_{k:\,x_k<c_k}(c_k-x_k)
\bigl[V_i^0(s+e_k)-V_i^0(s)\bigr],
\]
where $s+e_k$ is the state after one unit returns to lender
$k$. Both payoffs inside each difference lie between $0$
and $F_m$, so its absolute value is at most $F_m$.
There are at most $m+1$ outstanding units. Thus
\[
\left|
\bigl[(Q^\mu-Q^0)V_i^0\bigr](s)
\right|
\le(m+1)\mu F_m.
\]
As established in Step~1, the inverse of
$\delta\operatorname{Id}-Q^\mu$ has nonnegative entries and
row sums $\frac{1}{\delta}$. Consequently,
\[
\max_{s\in S}
\left|V_i^\mu(s)-V_i^0(s)\right|
\le
(m+1)F_m\frac{\mu}{\delta}
=:\varepsilon_\mu.
\]
This bound holds for every fixed pair of admissible
strategies.

\paragraph{Preserving the verified payoff gains.}
At a borrower arrival when both lenders are available,
switching between the two permitted fees changes a lender's
allocation probability by $\frac{1}{2}$. The continuation
part of its payoff gain is therefore one half of the
difference between its payoffs after lender 1 receives the
borrower and after lender 2 receives it. Each continuation
payoff changes by at most $\varepsilon_\mu$ when returns
are introduced, so the payoff gain changes by at most
$\varepsilon_\mu$. When a lender is alone, changing its
quote changes only the current fee, leaving the payoff
gain unaffected by returns.

A payoff gain from choosing one complete strategy rather
than another is the difference between the two resulting
payoffs from $s^0$. Since each changes by at most
$\varepsilon_\mu$, this gain changes by at most
$2\varepsilon_\mu$.

Using
\[
F_m=\frac{\beta(1-\beta^m)}{1-\beta},
\]
the choice of $\mu$ in the main proof gives
\[
\varepsilon_\mu
=
(m+1)\frac{\mu}{\delta}
\frac{\beta(1-\beta^m)}{1-\beta}
<
\frac{\beta\eta}{2}.
\]
Every verified gain at a borrower arrival consequently
remains at least
\[
\eta-\varepsilon_\mu>0,
\]
and each of the three complete-strategy gains identified
above remains greater than
\[
\beta\eta-2\varepsilon_\mu>0.
\]

\paragraph{Equilibrium verification.}
With $\lambda,\mu>0$, every state is reachable from $s^0$
under every admissible pair of strategies: successive
arrivals can lend all units before any return, and a
suitable sequence of returns before another arrival can
then produce any state. Together with the complete-strategy
identity in Step~1, this ensures that a positive quote
inequality at a state where a strategy changes gives a
strict payoff difference from $s^0$, provided the inequalities
at all other states are nonnegative.

Under AO, lender 1 therefore retains a unique best response
to each of lender 2's two constant quotes, including its
quote $1$ when alone. Lender 2 continues to prefer $\alpha$
against either response. Thus the strategies specified in
Step~2 form the unique pure stationary equilibrium under AO.

Under OA, lender 2's strategy from Step~3 remains a best
response. For lender 1, apply the continuation-payoff bound
separately to each strategy that fixes its common quote in
$(m,0)$ and $(m,1)$ to $\alpha$ or to $1$ and quotes $1$
whenever $1\le x_1<m$. Under either fixed choice, the positive
quote inequalities and the identity in Step~1 establish
optimality against every change to its quotes with $x_1<m$.
The positive payoff gain between the two resulting complete
strategies then selects $\alpha$ at $x_1=m$.

Both strategies are evaluated using their respective
matrices $Q^\mu$, which include returns from $x_1<m$ to
$x_1=m$. Lender 1's proposed strategy is therefore also a
best response, establishing the stated OA equilibrium.

\subsection{Equilibrium verification for Proposition~\ref{prop:two-fee-region}}
\label{app:two-fee-region}

Fix the informed lender's quotes when alone at $p_H$, as justified in
Appendix~\ref{app:theorem1}. This leaves eight strategy profiles
under AO and sixteen under OA. We evaluate each remaining strategy using \eqref{eq:payoff-system}. The resulting best responses are given below; the argument following the tables verifies them throughout the proposition's parameter region.

\medskip
\noindent\textbf{AO.}
Lender 1's unique best responses to lender 2's two possible quotes are
\begin{center}
\renewcommand{\arraystretch}{1.1}
\begin{tabular}{@{}ccc@{}}
\toprule
Lender 2's quote & \multicolumn{2}{c}{Lender 1's best response} \\
\cmidrule(l){2-3}
when $x_2=1$ & $s=(1,1)$ & $s=(2,1)$ \\
\midrule
$p_L$ & $p_L$ & $p_L$ \\
$p_H$ & $p_H$ & $p_L$ \\
\bottomrule
\end{tabular}
\end{center}
Against either response, lender 2 earns more by quoting $p_L$ than
$p_H$. Hence the unique AO equilibrium has lender 1 quote $p_L$ when
competing and $p_H$ when alone, while lender 2 always quotes $p_L$.

\medskip
\noindent\textbf{OA.}
Lender 2's unique best responses to all four lender-1 strategies are
\begin{center}
\renewcommand{\arraystretch}{1.1}
\begin{tabular}{@{}cccc@{}}
\toprule
\multicolumn{2}{c}{Lender 1's quotes}
 & \multicolumn{2}{c}{Lender 2's best response} \\
\cmidrule(r){1-2}\cmidrule(l){3-4}
$x_1=1$ & $x_1=2$ & $s=(1,1)$ & $s=(2,1)$ \\
\midrule
$p_L$ & $p_L$ & $p_L$ & $p_L$ \\
$p_L$ & $p_H$ & $p_L$ & $p_L$ \\
$p_H$ & $p_L$ & $p_H$ & $p_L$ \\
$p_H$ & $p_H$ & $p_H$ & $p_L$ \\
\bottomrule
\end{tabular}
\end{center}
Against lender 2's competitive quotes $(p_H,p_L)$, lender 1's unique
best response is $(p_H,p_L)$ at own supply 1 and 2. Against $(p_L,p_L)$,
either lender-1 strategy quoting $p_L$ at own supply 1 gains by raising
that quote to $p_H$. Only the third row therefore gives an OA equilibrium.

\noindent \paragraph{Exact verification.}
Normalize
\[
r=\frac{\lambda}{\mu},\qquad
d=\frac{\delta}{\mu},\qquad
\alpha=\frac{p_L}{p_H}.
\]
After dividing payoffs by $p_H$, every required payoff difference
has a strictly positive denominator independent of $\alpha$ and a
numerator affine in $\alpha$. It therefore suffices to verify the
endpoint values $\alpha=\frac{37}{50}$ and $\alpha=\frac{19}{25}$.
For either endpoint, write
\[
r=(1+d)\frac{52+55u}{40(1+u)},
\qquad d>0,\quad u\geq 0,
\]
which parameterizes
\[
\frac{13}{10}(1+d)
\leq r
<
\frac{11}{8}(1+d).
\]
After clearing positive denominators, every numerator is a polynomial
in $d$ and $u$ with nonnegative rational coefficients and is positive
at $u=0$ for every $d>0$. Hence every numerator is positive throughout
this half-open region. At the remaining upper boundary $r=\frac{11}{8}(1+d)$, every cleared numerator is a nonzero polynomial in $d$ with nonnegative rational coefficients, and is therefore
positive for $d>0$. Thus all twenty-five payoff comparisons are strict
throughout the stated parameter region, proving the displayed best
responses and the uniqueness of the $AO$ and $OA$ equilibria.

\subsection{Gains to both lenders}
\label{app:pareto}

Fix supplies $(c_1,c_2)=(3,1)$, permitted fees $P=\{0.75,1\}$,
$\lambda=2$, $\mu=1$ and $\delta=0.1$.

\paragraph{Regime AO}
Lender 2 chooses one quote whenever available, giving two
strategies. Lender 1 chooses quotes in states $(1,1)$, $(2,1)$
and $(3,1)$, giving eight remaining strategies.

\begin{center}
\setlength{\tabcolsep}{4pt}
\renewcommand{\arraystretch}{1.1}
{%
\begin{tabular}{@{}ccccrr@{}}
\toprule
\shortstack{Lender 2's quote\\when $x_2=1$}
& \multicolumn{3}{c}{Lender 1's best response}
& Payoff & \shortstack{Next-best\\payoff} \\
\cmidrule(lr){2-4}
& $(1,1)$ & $(2,1)$ & $(3,1)$ & & \\
\midrule
$\mathbf{0.75}$ & $\mathbf{0.75}$ & $\mathbf{0.75}$
& $\mathbf{0.75}$ & $\mathbf{11.7404}$ & $11.6288$ \\
$1$ & $1$ & $1$ & $0.75$ & $12.9418$ & $12.9240$ \\
\midrule
\shortstack{Lender 1's quotes\\in $(1,1),(2,1),(3,1)$}
& \multicolumn{3}{c}{Lender 2's best response}
& Payoff & \shortstack{Next-best\\payoff} \\
\cmidrule(lr){2-4}
& \multicolumn{3}{c}{$x_2=1$} & & \\
\midrule
$\boldsymbol{(0.75,0.75,0.75)}$
& \multicolumn{3}{c}{$\mathbf{0.75}$}
& $\mathbf{4.0717}$ & $2.1340$ \\
$(1,1,0.75)$
& \multicolumn{3}{c}{$0.75$}
& $4.8244$ & $4.4666$ \\
\bottomrule
\end{tabular}}
\end{center}

Every displayed best response is unique. Lender 2 responds to
either lender-1 strategy in the second block with $0.75$.
Lender 1 then responds by quoting $0.75$ in all three states
where both lenders are available. Hence the unique equilibrium
$\sigma^{\mathrm{AO}}$ has lender 1 quoting $1$ when alone and
$0.75$ when competing, while lender 2 quotes $0.75$ whenever
available.

\paragraph{Regime OA}
Lender 1 chooses quotes at own available supplies $x_1=1,2,3$.
Lender 2 chooses quotes in states $(1,1)$, $(2,1)$ and $(3,1)$.
Each lender has eight remaining strategies.

\begin{center}
\setlength{\tabcolsep}{4pt}
\renewcommand{\arraystretch}{1.1}{%
\begin{tabular}{@{}ccccrr@{}}
\toprule
\shortstack{Lender 1's quotes\\at $x_1=1,2,3$}
& \multicolumn{3}{c}{Lender 2's best response}
& Payoff & \shortstack{Next-best\\payoff} \\
\cmidrule(lr){2-4}
& $(1,1)$ & $(2,1)$ & $(3,1)$ & & \\
\midrule
$(0.75,0.75,0.75)$ & $0.75$ & $0.75$ & $0.75$ & $4.2538$ & $4.0237$ \\
$(0.75,0.75,1)$    & $0.75$ & $0.75$ & $0.75$ & $4.9581$ & $4.7678$ \\
$(0.75,1,0.75)$    & $0.75$ & $0.75$ & $0.75$ & $4.8139$ & $4.7033$ \\
$(0.75,1,1)$       & $0.75$ & $0.75$ & $0.75$ & $5.3121$ & $5.2819$ \\
$(1,0.75,0.75)$    & $1$    & $0.75$ & $0.75$ & $4.5114$ & $4.4362$ \\
$(1,0.75,1)$       & $1$    & $0.75$ & $0.75$ & $5.1710$ & $5.1088$ \\
$(1,1,0.75)$       & $1$    & $1$    & $0.75$ & $4.9609$ & $4.9502$ \\
$\boldsymbol{(1,1,1)}$ & $\mathbf{1}$ & $\mathbf{1}$
& $\mathbf{0.75}$ & $\mathbf{5.4552}$ & $5.4484$ \\
\midrule
\shortstack{Lender 2's quotes\\in $(1,1),(2,1),(3,1)$}
& \multicolumn{3}{c}{Lender 1's best response}
& Payoff & \shortstack{Next-best\\payoff} \\
\cmidrule(lr){2-4}
& $x_1=1$ & $x_1=2$ & $x_1=3$ & & \\
\midrule
$(0.75,0.75,0.75)$ & $1$ & $1$ & $1$ & $11.2477$ & $10.7906$ \\
$(1,0.75,0.75)$    & $1$ & $1$ & $1$ & $11.4411$ & $10.9925$ \\
$\boldsymbol{(1,1,0.75)}$ & $\mathbf{1}$ & $\mathbf{1}$
& $\mathbf{1}$ & $\mathbf{11.9532}$ & $11.8808$ \\
\bottomrule
\end{tabular}}
\end{center}

Every displayed best response is unique. Lender 1 responds to
each of lender 2's three strategies in the second block by
quoting $1$ whenever available. Lender 2 then quotes $0.75$
only in $(3,1)$ and $1$ in every other state where it is
available. This gives the unique equilibrium
$\sigma^{\mathrm{OA}}$. The exact payoff differences between the two equilibria satisfy
\[
\begin{aligned}
V_1(s^0,\sigma^{\mathrm{OA}})
-V_1(s^0,\sigma^{\mathrm{AO}})&>0.2127,\\
V_2(s^0,\sigma^{\mathrm{OA}})
-V_2(s^0,\sigma^{\mathrm{AO}})&>1.3834.
\end{aligned}
\]
Thus both lenders strictly prefer OA to AO.

\section{Ownership Fragmentation: Equilibrium Verification} \label{app:fragmentation}
This appendix verifies that the strategies in Proposition~\ref{prop:split-informed-lender} form the unique pure stationary equilibrium in each game throughout the stated parameter region. The total-payoff comparison is proved in the main text.

\paragraph{Remaining strategies.}
Fix informed lenders' fees at $p_H$ when they observe zero aggregate rival available supply, as shown in Section~\ref{c:sole_lender}. Every equilibrium
must satisfy this restriction. 

\noindent In the $(2,2)$ game under OA, lender 1 chooses its fees at $x_1=1,2$, giving four strategies. Lender 2 chooses its fees at observations $(x_2,x_1)\in\{1,2\}\times\{1,2\}$, giving sixteen remaining strategies. There are therefore $4\cdot16=64$ remaining strategy combinations.

\noindent In the $(2,1,1)$ game under OAA, lender 1 again has four strategies. Each informed lender has one unit and chooses its fee at each positive aggregate rival available supply in $\{1,2,3\}$, giving eight remaining strategies. These fees cannot distinguish states with the same aggregate rival available supply. There are therefore $4\cdot8\cdot8=256$ remaining strategy combinations. The two informed lenders need not use the same strategy.

\paragraph{Payoff comparisons.}
Write $r=\frac{\lambda}{\mu}$, $d=\frac{\delta}{\mu}$ and $\alpha=\frac{p_L}{p_H}$.

The parameter region becomes
\[
0<d\leq1,\qquad 4+4d\leq r\leq\frac{21}{5}+4d, \qquad \frac{899}{1000}\leq\alpha\leq\frac{901}{1000}.
\]
For each game, let $\sigma^*$ denote the strategies in Proposition~\ref{prop:split-informed-lender}. For an alternative strategy $\widetilde\sigma_i$ and fixed rivals' strategies $\sigma_{-i}$, compare
\[
\Delta_i(\widetilde\sigma_i,\sigma_{-i}) = \frac{1}{p_H}\bigl[ V_i(s^0,(\sigma_i^*,\sigma_{-i}))-V_i(s^0,(\widetilde\sigma_i,\sigma_{-i}))\bigr].
\]
For each informed lender, we check every alternative against every combination of rivals' remaining strategies. We then check lender 1's three alternatives against the informed lenders' proposed strategies. We compute the normalised payoffs from~\eqref{eq:payoff-system} using borrower arrival rate $r$, return rate $1$, discount rate $d$, and fees $\{\alpha,1\}$. Each calculation uses the strategies in that payoff. The numbers of comparisons are
\begin{center}
\small
\setlength{\tabcolsep}{3pt}
\begin{tabular}{@{}ccrrr@{}}
\toprule
Regime & Lender & \shortstack{Rivals'\\strategies}
& \shortstack{Alternative\\strategies} & Comparisons \\
\midrule
OA  & 2 & 4  & 15 & 60 \\
OA  & 1 & 1  & 3  & 3 \\
OAA & 2 & 32 & 7  & 224 \\
OAA & 1 & 1  & 3  & 3 \\
\bottomrule
\end{tabular}
\end{center}
In OAA, exchanging lenders 2 and 3 gives the same comparisons for lender 3. This leaves 290 payoff differences to verify, without restricting the two informed lenders to identical strategies.

\paragraph{Verification over the parameter region.}
For a fixed choice of the lower or higher fee at each observation, $Q(\sigma)$ is independent of $\alpha$. Equation~\eqref{eq:payoff-system} therefore makes each payoff difference rational in $(r,d)$ and affine in $\alpha$. The computed denominators have nonnegative coefficients in $r,d$ and are nonzero, so they are positive for $r,d>0$. It suffices to check
\[
\alpha\in\left\{\frac{899}{1000},\frac{901}{1000}\right\}.
\]
At either endpoint, substitute
\[
d=\frac{1}{1+u},\qquad
r=4+4d+\frac{v}{5(1+v)},\qquad u,v\geq0.
\]
This covers $0<d\leq1$ and $4+4d\leq r<\frac{21}{5}+4d$. At the remaining upper boundary, substitute $d=\frac{1}{1+u}$ and $r=\frac{21}{5}+4d,\qquad u\geq0$.
After clearing positive denominators, every resulting numerator has nonnegative coefficients and a positive constant term. These checks were performed using exact arithmetic for all 290 payoff differences. Thus all 1160 polynomial checks are satisfied, and every required payoff difference is strictly positive throughout the parameter region. In each game, every informed lender's proposed strategy is therefore its unique best response to every combination of rivals' remaining strategies. Every equilibrium must use these informed strategies. Against them, lender 1's proposed strategy is its unique best response. The proposed strategies are consequently mutual best responses, and no other combination can be an equilibrium.


\end{document}

%% file: figure_arxiv.tex
%

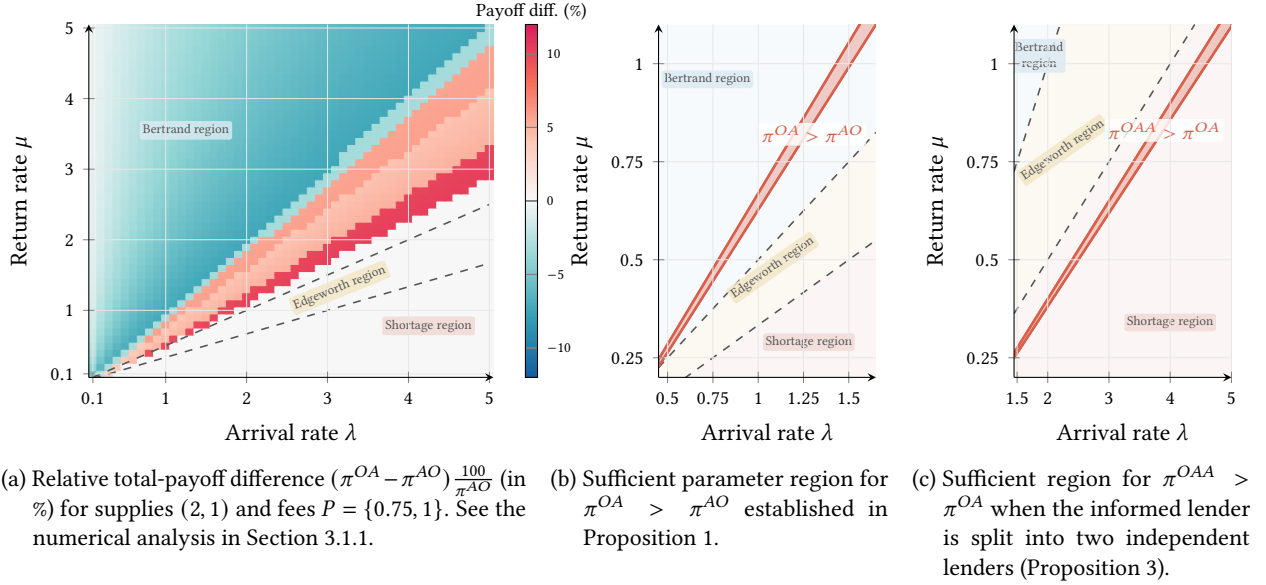
\begin{figure*}
\centering

\definecolor{HeatNegDark}{HTML}{145DA0}
\definecolor{HeatNegMid}{HTML}{329FB4}
\definecolor{HeatNegLight}{HTML}{93D5D5}
\definecolor{HeatZero}{HTML}{F7F7F7}
\definecolor{HeatPosLight}{HTML}{FDC7B5}
\definecolor{HeatPosMid}{HTML}{F76B57}
\definecolor{HeatPosDark}{HTML}{D81B60}
\definecolor{NoEqGrey}{HTML}{C7C7C7}
\definecolor{InkGrey}{HTML}{555555}
\definecolor{GridGrey}{HTML}{E7E7E7}
\definecolor{BertrandShade}{HTML}{DCEAF3}
\definecolor{EdgeworthShade}{HTML}{F3E7C7}
\definecolor{ShortageShade}{HTML}{F1D9D5}
\definecolor{ProofWarm}{HTML}{D95F4B}

\pgfplotsset{
  figone axis/.style={
    axis lines=left,
    tick label style={font=\scriptsize},
    label style={font=\small},
    grid=major,
    grid style={draw=GridGrey, line width=0.25pt},
    axis on top,
    clip=true,
  },
  capacity boundary/.style={InkGrey, dashed, line width=0.6pt},
  economic shade/.style={draw=none, fill opacity=0.12},
  economic shade plain/.style={draw=none, fill opacity=0.23},
  proved band/.style={draw=ProofWarm, fill=ProofWarm, fill opacity=0.30,
                      draw opacity=1, line width=0.9pt},
}

\pgfplotsset{
  colormap={payoffmap}{
    rgb255(0cm)=(20,93,160);
    rgb255(1cm)=(50,159,180);
    rgb255(2cm)=(147,213,213);
    rgb255(3cm)=(247,247,247);
    rgb255(4cm)=(253,199,181);
    rgb255(5cm)=(247,107,87);
    rgb255(6cm)=(216,27,96)
  }
}

\begin{subfigure}[t]{0.42\textwidth}
\centering
\begin{tikzpicture}
\begin{axis}[
    figone axis,
    width=\linewidth,
    height=6.25cm,
    xmin=0.05, xmax=5.05,
    ymin=0.05, ymax=5.05,
    xlabel={Arrival rate $\lambda$},
    ylabel={Return rate $\mu$},
    xtick={0.1,1,2,3,4,5},
    ytick={0.1,1,2,3,4,5},
    axis background/.style={fill=white},
    colorbar,
    colorbar style={
      width=1.6mm,
      ytick={-10,-5,0,5,10},
      title={\scriptsize Payoff diff. (\%)},
      title style={yshift=-2.8mm, font=\tiny},
      tick label style={font=\tiny},
    },
    point meta min=-12,
    point meta max=12,
]
\addplot[
    matrix plot*,
    mesh/cols=50,
    point meta=explicit,
    colormap name=payoffmap,
    unbounded coords=discard,
] table[x=lambda,y=mu,meta=diff,col sep=comma]
    {fig1a_heatmap_recovered_long.csv};

\addplot[name path=aBottom, draw=none, domain=0.05:5.05, samples=2] {0.05};
\addplot[name path=aCtot, draw=none, domain=0.05:5.05, samples=2] {x/3};
\addplot[name path=aCmax, draw=none, domain=0.05:5.05, samples=2] {x/2};
\addplot[name path=aTop, draw=none, domain=0.05:5.05, samples=2] {5.05};

\addplot[economic shade, fill=ShortageShade] fill between[of=aBottom and aCtot];
\addplot[economic shade, fill=EdgeworthShade] fill between[of=aCtot and aCmax];
\addplot[economic shade, fill=BertrandShade] fill between[of=aCmax and aTop];

\addplot[capacity boundary, domain=0.10:5.0, samples=2] {x/2};
\addplot[capacity boundary, domain=0.10:5.0, samples=2] {x/3};

\node[font=\tiny, text=InkGrey, fill=BertrandShade, fill opacity=0.82, text opacity=1,
      rounded corners=1pt, inner sep=1.2pt]
    at (axis cs:1.25,3.55) {Bertrand region};
\node[font=\tiny, text=InkGrey, fill=EdgeworthShade, fill opacity=0.84, text opacity=1,
      rounded corners=1pt, inner sep=1.2pt, rotate=22]
    at (axis cs:3.15,1.32) {Edgeworth region};
\node[font=\tiny, text=InkGrey, fill=ShortageShade, fill opacity=0.84, text opacity=1,
      rounded corners=1pt, inner sep=1.2pt]
    at (axis cs:4.25,0.78) {Shortage region};

\end{axis}
\end{tikzpicture}
\caption{Relative total-payoff difference
$(\pi^{OA}-\pi^{AO}) \frac{100}{\pi^{AO}}$ (in \%)
for supplies $(2,1)$ and fees $P=\{0.75,1\}$.
See the numerical analysis in
Section~\ref{par:numerical-comparison}.}
\label{fig:payoff-map-regimes}
\end{subfigure}\hfill
\begin{subfigure}[t]{0.27\textwidth}
\centering
\begin{tikzpicture}
\begin{axis}[
    figone axis,
    width=\linewidth,
    height=6.25cm,
    xmin=0.45, xmax=1.65,
    ymin=0.20, ymax=1.10,
    xlabel={Arrival rate $\lambda$},
    ylabel={Return rate $\mu$},
    xtick={0.5,0.75,1,1.25,1.5},
    ytick={0.25,0.5,0.75,1},
]
\addplot[name path=bBottom, draw=none, domain=0.45:1.65, samples=2] {0.20};
\addplot[name path=bCtot, draw=none, domain=0.45:1.65, samples=2] {x/3};
\addplot[name path=bCmax, draw=none, domain=0.45:1.65, samples=2] {x/2};
\addplot[name path=bTop, draw=none, domain=0.45:1.65, samples=2] {1.10};
\addplot[economic shade plain, fill=ShortageShade] fill between[of=bBottom and bCtot];
\addplot[economic shade plain, fill=EdgeworthShade] fill between[of=bCtot and bCmax];
\addplot[economic shade plain, fill=BertrandShade] fill between[of=bCmax and bTop];
\addplot[capacity boundary, domain=0.45:1.65, samples=2] {x/2};
\addplot[capacity boundary, domain=0.45:1.65, samples=2] {x/3};

\addplot[name path=pLower, draw=none, domain=0.45:1.65, samples=2]
    {x/1.375 - 0.1};
\addplot[name path=pUpper, draw=none, domain=0.45:1.65, samples=2]
    {x/1.30 - 0.1};
\addplot[proved band] fill between[of=pLower and pUpper];
\addplot[ProofWarm, line width=0.9pt, domain=0.45:1.65, samples=2]
    {x/1.375 - 0.1};
\addplot[ProofWarm, line width=0.9pt, domain=0.45:1.65, samples=2]
    {x/1.30 - 0.1};

\node[font=\tiny, text=InkGrey, fill=BertrandShade, fill opacity=0.82, text opacity=1,
      rounded corners=1pt, inner sep=1.0pt]
    at (axis cs:0.72,0.96) {Bertrand region};
\node[font=\tiny, text=InkGrey, fill=EdgeworthShade, fill opacity=0.84, text opacity=1,
      rounded corners=1pt, inner sep=1.0pt, rotate=35]
    at (axis cs:1.08,0.48) {Edgeworth region};
\node[font=\tiny, text=InkGrey, fill=ShortageShade, fill opacity=0.82, text opacity=1,
      rounded corners=1pt, inner sep=1.0pt]
    at (axis cs:1.28,0.29) {Shortage region};
\node[font=\scriptsize, text=ProofWarm, fill=white, fill opacity=0.82, text opacity=1,
      rounded corners=1pt, inner sep=1.1pt]
    at (axis cs:1.30,0.82) {$\pi^{OA}>\pi^{AO}$};
\end{axis}
\end{tikzpicture}
\caption{Sufficient parameter region for
$\pi^{OA}>\pi^{AO}$ established in
Proposition~\ref{prop:two-fee-region}.}
\label{fig:prop1-region}
\end{subfigure}\hfill
\begin{subfigure}[t]{0.27\textwidth}
\centering
\begin{tikzpicture}
\begin{axis}[
    figone axis,
    width=\linewidth,
    height=6.25cm,
    xmin=1.45, xmax=5.00,
    ymin=0.20, ymax=1.10,
    xlabel={Arrival rate $\lambda$},
    ylabel={Return rate $\mu$},
    xtick={1.5,2,3,4,5},
    ytick={0.25,0.5,0.75,1},
]
\addplot[name path=cBottom, draw=none, domain=1.45:5.00, samples=2] {0.20};
\addplot[name path=cCtot, draw=none, domain=1.45:5.00, samples=2] {x/4};
\addplot[name path=cCmax, draw=none, domain=1.45:5.00, samples=2] {x/2};
\addplot[name path=cTop, draw=none, domain=1.45:5.00, samples=2] {1.10};
\addplot[economic shade plain, fill=ShortageShade] fill between[of=cBottom and cCtot];
\addplot[economic shade plain, fill=EdgeworthShade] fill between[of=cCtot and cCmax];
\addplot[economic shade plain, fill=BertrandShade] fill between[of=cCmax and cTop];
\addplot[capacity boundary, domain=1.45:2.20, samples=2] {x/2};
\addplot[capacity boundary, domain=1.45:4.40, samples=2] {x/4};

\addplot[name path=sLower, draw=none, domain=1.45:5.00, samples=2]
    {(x-0.4)/4.2};
\addplot[name path=sUpper, draw=none, domain=1.45:5.00, samples=2]
    {x/4 - 0.1};
\addplot[proved band] fill between[of=sLower and sUpper];
\addplot[ProofWarm, line width=0.9pt, domain=1.45:5.00, samples=2]
    {(x-0.4)/4.2};
\addplot[ProofWarm, line width=0.9pt, domain=1.45:5.00, samples=2]
    {x/4 - 0.1};

\node[font=\tiny, text=InkGrey, fill=BertrandShade, fill opacity=0.82, text opacity=1,
      rounded corners=1pt, inner sep=1.0pt, align=center]
    at (axis cs:1.86,1.02) {Bertrand\\region};
\node[font=\tiny, text=InkGrey, fill=EdgeworthShade, fill opacity=0.84, text opacity=1,
      rounded corners=1pt, inner sep=1.0pt, rotate=35]
    at (axis cs:2.25,0.78) {Edgeworth region};
\node[font=\tiny, text=InkGrey, fill=ShortageShade, fill opacity=0.82, text opacity=1,
      rounded corners=1pt, inner sep=1.0pt]
    at (axis cs:4.00,0.34) {Shortage region};
\node[font=\scriptsize, text=ProofWarm, fill=white, fill opacity=0.82, text opacity=1,
      rounded corners=1pt, inner sep=1.1pt]
    at (axis cs:3.90,0.83) {$\pi^{OAA}>\pi^{OA}$};
\end{axis}
\end{tikzpicture}
\caption{Sufficient region for
$\pi^{OAA}>\pi^{OA}$ when the informed lender
is split into two independent lenders
(Proposition~\ref{prop:split-informed-lender}).}
\label{fig:splitting-region}
\end{subfigure}

\caption{Payoff comparisons across borrower arrival rates
$\lambda$ and return rates $\mu$, with discount rate
$\delta=0.1$ throughout.
The dashed lines distinguish three demand regimes:
the Bertrand region ($\lambda\leq c_{\max}\mu$),
the Edgeworth region ($c_{\max}\mu<\lambda\leq C\mu$),
and the shortage region ($\lambda>C\mu$),
where $c_{\max}=\max_{i\in I}c_i$ is the largest
individual lender's supply and $C=\sum_{i\in I}c_i$
is aggregate supply.
Panel~(a) presents numerical results, while
panels~(b) and~(c) illustrate the analytically established
sufficient regions from
Propositions~\ref{prop:two-fee-region} and
\ref{prop:split-informed-lender}, respectively.
The shaded demand regimes provide economic benchmarks
rather than equilibrium boundaries.}
\label{fig:main-regime-map}
\end{figure*}